\documentclass[11pt]{amsart}

\usepackage{graphicx}
\usepackage{framed}
\IfFileExists{ulem.sty}{\usepackage{ulem}}{}
\usepackage{pgf,tikz,pgfplots}
\usepackage{float}
\usepackage{mathrsfs}
\usepackage{mathtools}
\usepackage{mathpple}
\usepackage[all,cmtip]{xy}
\usepackage{amsmath}
\usepackage{tikz-cd}
\IfFileExists{mathdots.sty}{\usepackage{mathdots}}{}
\IfFileExists{yhmath.sty}{\usepackage{yhmath}}{}
\usepackage{cancel}
\usepackage{color}
\definecolor{deepgreen}{rgb}{0.0,0.35,0.0}
\IfFileExists{siunitx.sty}{\usepackage{siunitx}}{}
\usepackage{array}
\usepackage{multirow}
\usepackage{amssymb}
\usepackage{gensymb}
\usepackage{tabularx}
\usepackage{booktabs}
\usepackage{makecell}
\usetikzlibrary{fadings}
\usetikzlibrary{patterns}
\usetikzlibrary{shapes}
\usetikzlibrary{arrows.meta,calc,decorations.markings}
\usetikzlibrary{matrix,arrows,decorations.pathmorphing}
\usepackage{enumerate}
\usepackage{enumitem}
\IfFileExists{extarrows.sty}{\usepackage{extarrows}}{}
\usepackage{appendix}
\usepackage[bookmarks=true,bookmarksopen=true,hidelinks]{hyperref}

\numberwithin{equation}{section}
\allowdisplaybreaks[2]

\newtheorem{theorem}{Theorem}[section]
\newtheorem{proposition}[theorem]{Proposition}
\newtheorem{lemma}[theorem]{Lemma}
\newtheorem{corollary}[theorem]{Corollary}

\newtheorem*{theoremA}{Theorem A}
\newtheorem*{theoremB}{Theorem B}
\newtheorem*{theoremC}{Theorem C}
\theoremstyle{definition}
\newtheorem{definition}[theorem]{Definition}
\newtheorem{example}[theorem]{Example}
\theoremstyle{remark}
\newtheorem{remark}[theorem]{Remark}

\newcommand{\C}{\mathbb C}
\newcommand{\R}{\mathbb R}
\newcommand{\Z}{\mathbb Z}

\newcommand{\Disc}{\operatorname{Disc}}

\newcommand{\Res}{\operatorname*{Res}}
\newcommand{\sech}{\operatorname{sech}}

\newcommand{\dd}{\,\mathrm{d}}

\newcommand{\CritVal}{\operatorname{CritVal}}
\newcommand{\Sym}{\operatorname{Sym}}

\newcommand{\re}{\operatorname{Re}}
\newcommand{\im}{\operatorname{Im}}

\setlist[enumerate]{itemsep=3pt,topsep=5pt}

\title{Resurgence and Complex Geodesics for the Heisenberg Heat Kernel}
\author{Xinxing Tang}
\thanks{X.~Tang: Beijing Institute of Mathematical Sciences and Applications,
Beijing, China.
Email: \texttt{tangxinxing@bimsa.cn}.}
\date{}
\hypersetup{pdftitle={Resurgence and Complex Geodesics for the Heisenberg Heat Kernel}}
\begin{document}

\begin{abstract}
We study complex geodesic actions and heat-kernel resurgence on isotropic Heisenberg groups. For real endpoints with $r>0,y\ne0$, the nonzero finite singularities reached by analytic continuation of the minimizing-saddle Borel germ are exactly the finite complex action differences. For nonresonant endpoints, we construct row-finite directional Stokes operators and compute the positive-ray Stokes matrix and its logarithm explicitly; the primitive alien support can be strictly smaller than the full Borel singular support.

Near the vertical axis, we obtain uniform two-term large-order asymptotics from the nearest Borel singularity. Cubic resonances admit uniform $A_2$ confluence formulas. We also identify the combination of resurgent sectors selected by the original Fourier contour. Positive Borel summation is exact at horizontal endpoints. At nonzero isotropic complex endpoints, no finite normal trajectories exist, but the signed vertical action lattice is the complete limiting spectrum. The action and Stokes structures are independent of the dimension on isotropic $H^n$.
\end{abstract}

\maketitle

\tableofcontents

\section{Introduction}\label{sec:introduction}

Minimizing trajectories govern the short-time heat-kernel expansion, while its Borel singularities may also detect nonminimizing complex geodesic actions. On the Heisenberg group, Fourier transformation in the central variable reduces the heat kernel to a one-dimensional exponential integral whose critical values encode these actions. We study both the continuation of the resulting Borel germ and the reconstruction of the physical Fourier contour.

\subsection{Geometric question and Heisenberg model}

Let $(M,g)$ be a real-analytic Riemannian manifold with a suitable complexification $(M_{\C},g_{\C})$, and suppose that $p_0,p_1$ are joined by a unique nonconjugate minimizing geodesic. The heat kernel has the short-time expansion
\[
K(p_0,p_1;t)\sim(2\pi t)^{-\dim M/2}e^{-A_0/t}\sum_{n\ge0}c_nt^n,
\qquad A_0=\frac12\operatorname{dist}(p_0,p_1)^2.
\]
Kontsevich \cite{Kontsevich2020} proposed that the resurgence of the series after removal of the minimizing exponential should reflect the actions
\[
A_\gamma=\frac12\int_0^1g_{\C}(\dot\gamma,\dot\gamma)\,\dd s
\]
of complex geodesics $\gamma$ joining the same endpoints; see also \cite{KS2022,KS2024,LLTheat}. The expected Borel singularities are the action differences
\[A_\gamma-A_0.\]

For finite-dimensional exponential integrals, the relation between saddle geometry, vanishing cycles, Stokes phenomena and resurgence is classical; see \cite{BerryHowls,DelabaerePham,DelabaereHowls,KS2024,Pham1983}. Our previous work \cite{LLTalien,LLTheat} concerns Picard--Lefschetz theory, alien calculus, and heat-kernel resurgence. Here the main analytic problem is to continue the inverse branches of a nonproper meromorphic phase. A critical value is singular for a continued Borel germ only if the transported zero-cycle has a nonzero local component.

\vskip 0.1cm
We study these questions for the Heisenberg sub-Laplacian
\[
X_1=\partial_{x_1}+x_2\partial_y,\qquad
X_2=\partial_{x_2}-x_1\partial_y,\qquad
\mathcal L_H=\frac12(X_1^2+X_2^2).
\]
The central coordinate records signed stochastic area \cite{BBBC}. Fourier transformation in this coordinate gives a magnetic Mehler kernel and reduces the problem to a single complex phase. The classical Fourier representation is due to Gaveau and Hulanicki \cite{Gaveau,Hulanicki}, while its relation with complex Hamiltonian mechanics was developed by Beals--Gaveau--Greiner \cite{BGG}. Off-cut-locus asymptotics, conjugacy, and the near-axis Bessel transition were also established in \cite{BBN,BenArous,CL,Li2007,LZ}.

\vskip 0.1cm
We determine the complex action spectrum, Borel singularities, Stokes structure, and large-order asymptotics of the Heisenberg heat kernel, and relate the resurgent sectors to the original Fourier contour. We also treat horizontal and isotropic complex endpoints.

For a real endpoint, set $a=r^2/2=(x_1^2+x_2^2)/2$. The Heisenberg heat kernel based at the identity is given by
\begin{equation}\label{eq:hfourier}
p_t(r,y)=\frac1{(2\pi t)^2}\int_{\R}e^{-f(\tau)/t}V(\tau)\,\dd\tau,
\end{equation}
where
\[f(\tau)=-iy\tau+a\tau\coth\tau,\qquad  V(\tau)=\frac{\tau}{\sinh\tau}.\]
Under the change of variable $\tau=iq$, the phase and amplitude become
\[
F(q)=f(iq)=yq+aq\cot q,\qquad
W(q)=V(iq)=\frac{q}{\sin q}.\]

We write
\[
\Sigma_H^{\rm fin}(r,y)
\]
for the action spectrum of finite complex normal geodesics reaching the fixed endpoint. For $r,y>0$, we denote the minimal action by $A_1$. The positive critical values are locally finite. Away from a discrete set of parameters they are Morse; when
\[
C:=y/a=\theta_j,\qquad\text{where }~ \tan\theta_j=\theta_j,
\]
two adjacent Morse critical points coalesce into a cubic critical point. 

By the symmetry $y\mapsto -y$, the case $y<0$ reduces to $y>0$. We distinguish the following real endpoint regimes:
\[
\left\{
\begin{array}{ll}
r>0,\ y>0 &
\left\{
\begin{array}{ll}
C\neq \theta_j & \text{nonresonant},\\
C=\theta_j & \text{resonant},
\end{array}
\right.\\[1.2em]
r>0,\ y=0 & \text{horizontal boundary},\\[0.3em]
r=0,\ y>0 & \text{vertical boundary}.
\end{array}
\right.
\]
We also consider the isotropic complex endpoints:
\[
x_*=(x_{*1},x_{*2})\in\C^2\setminus\{0\},\qquad \sigma_*:=x_{*1}^2+x_{*2}^2=0,\qquad y>0.
\]
These endpoints have the same quadratic invariant $\sigma$ as a vertical endpoint, but admit no finite normal trajectories. The heat kernel still extends analytically, and the vertical action lattice remains as the limiting spectrum.

\subsection{Main results}
The oriented Fourier contour determines two local inverse branches at the minimizing saddle. With $\xi=X-A_1$ in the action variable, their pushed density has the form
\[
\rho(\xi)=\sum_{n\ge0}r_n\xi^{n-1/2}.
\]
We set
\begin{equation}\label{eq:Stry}
S(t)=\sum_{n\ge0}r_n\Gamma(n+1/2)t^n,\qquad
p_t(r,y)\sim\frac{e^{-A_1/t}}{4\pi^2t^{3/2}}S(t),
\end{equation}
and write $\widehat S(\xi)=\sum r_n\Gamma(n+1/2)\xi^n/n!$ for the ordinary Borel transform. The paths $\Gamma_\pm$ follow the positive action axis and pass each subsequent real critical value above or below. For a continuation path $\mathfrak p$, let $\operatorname{Sing}^*_{\mathfrak p}$ denote the nonzero finite singularities of the Borel germ continued along that path.

\begin{theoremA}[Global action-spectrum resurgence; Theorems~\ref{thm:hpositivespectrum}, \ref{thm:hfullspectrum},
\ref{thm:hhorizontalresurgence}, and \ref{thm:hhigherdim}]
For every $r,y>0$,
\[
\operatorname{Sing}_{\Gamma_\pm}\widehat S=\bigl(\Sigma_H^{\rm fin}(r,y)-A_1\bigr)\cap(0,\infty),\qquad
\bigcup_{\mathfrak p}\operatorname{Sing}^*_{\mathfrak p}\widehat S
=\bigl(\Sigma_H^{\rm fin}(r,y)-A_1\bigr)\setminus\{0\}.
\]
Every noninitial Morse action occurs with a nonzero logarithmic ordinary Borel singularity. At $C=\theta_j$, the cubic action occurs with exponent
$-1/6$. The same full-spectrum statement holds at horizontal endpoints after the even-symmetry quotient, and in every fixed isotropic Heisenberg dimension.
\end{theoremA}

The labelled saddle sectors also carry a global Stokes structure. 

\begin{theoremB}[Global Stokes structure;
Theorems~\ref{thm:directional}, \ref{thm:positive}, and
\ref{thm:angular}]
Fix a nonresonant real endpoint. Every affine line in the action plane meets $\Sigma_H^{\rm fin}(r,y)$ in a finite set. Hence, for every direction $\theta$, the labelled saddle sectors satisfy
the row-finite Stokes relation
\[
\Psi_A^{\theta,+}=\Psi_A^{\theta,-}
+\sum_{\substack{B\in\Sigma_H^{\rm fin}(r,y)\\
B-A\in e^{i\theta}\R_+}}s_{AB}^{\theta}\Psi_B^{\theta,-},
\qquad s_{AB}^{\theta}\in\{-2,-1,0,1,2\}.
\]
If
\[
\theta_m<C<\theta_{m+1},\qquad
A_1<B_1<A_2<\cdots<B_m<A_{m+1},
\]
then, in the integral-cycle normalization, the complete positive-ray Stokes matrix is
\[
(S_0)_{ij}=
\begin{cases}
  0, \quad & i>j,\\
  1, & i=j,\\
  2, & (i,j)=(2k-1,2k),\quad 1\le k\le m,\\
  1, & i<j ~\text{ otherwise}.
\end{cases}
\]
Moreover, the angle-ordered Stokes product is well defined on every closed angular cone of aperture $<\pi$ and admits a weighted completion on the labelled saddle module.
\end{theoremB}

At resonance, two neighboring Morse critical points coalesce into a cubic $A_2$ block. The uniform confluence formulas give action separation of order $\Lambda^{3/2}$ and universal scaled Borel profiles.

\vskip 0.1cm
Near the vertical axis, the nearest singularity of the initial Borel germ gives the following uniform large-order asymptotics.

\begin{theoremC}[Uniform large order near the vertical axis; Theorem~\ref{thm:hlargeorder}]
Let
\[
B_1=F(q_+),\qquad\omega_1=B_1-A_1,\qquad
\chi_1=-\frac{\sqrt2\,W(q_+)}{\sqrt{F''(q_+)}}>0,
\]
where $q_+$ is the first critical point to the right of $\pi$. For $y\in K\Subset(0,\infty)$ and $0<r<r_0(K)$, $\omega_1$ is the unique nearest nonzero singularity of the initial Borel germ, and
\[
\widehat S(\xi)=-\frac{\chi_1}{\sqrt\pi}\log\!\left(1-\frac{\xi}{\omega_1}\right)
+\text{less singular terms}.
\]
Moreover,
\[
c_n=\frac{\chi_1}{\sqrt\pi}\Gamma(n)\omega_1^{-n}
\left(1+\frac{\beta(r,y)}{n}+O_K(n^{-2})\right),
\]
and $\beta(r,y)\rightarrow-1/4$ uniformly as $r\rightarrow0^+$.
\end{theoremC}

%The Heisenberg group provides a particularly effective model for studying the relation between complex geometry and heat-kernel resurgence. After Fourier transformation in the central variable, the heat kernel reduces to a one-dimensional meromorphic phase, while still retaining nontrivial complex normal geodesics, Morse and cubic degenerations, and Stokes phenomena. This makes it possible to follow explicitly the passage from complex-geodesic actions to phase critical values, inverse-sheet ramification, Borel singularities, and Stokes data, and also to compare the intrinsic resurgent structure with the combination selected by the original Fourier contour. In this sense, the Heisenberg group serves as a controlled model for understanding which geometric data are detected by heat-kernel resurgence and which depend on global analytic continuation.
The Heisenberg model thus gives a direct comparison between complex-geodesic actions, Borel singularities, and Stokes data.

\subsection{Organization}
In Section \ref{sec:hsetup}, we relate the complex normal actions to the Fourier phase. In Section \ref{sec:hglobal}, we prove the Borel continuation and action-spectrum results, and furthermore, in Section \ref{sec:hglobalstokes}, we construct the Stokes operators. In Section \ref{sec:hconnection}, we reconstruct the original Fourier contour. In Sections \ref{sec:hboundary} and \ref{sec:hhigherdim}, we study the boundary regimes and the higher-dimensional extension.
\subsubsection*{Acknowledgments}

The author thanks Yong Li for early discussions and Ryszard Nest for comments on the first manuscript. X.~T. was supported by BNSFC Youth Project No.1254041, BNSFC No.JR25001 and NSFC Youth Project No.12501079. 

\section{Fourier phase and complex actions}\label{sec:hsetup}

We recall the complex Hamiltonian description \cite{BGG} in the normalization used here. For an endpoint $(x,y)\in\C^2\times\C$, set 
\[\sigma=x_1^2+x_2^2,\qquad a=\sigma/2,\qquad C=y/a ~\text{ when }a\ne0\] 
and
\[F(q)=yq+\frac{\sigma}{2}q\cot q,\qquad
\mathscr H(q)=\frac{q-\sin q\cos q}{\sin^2q}.
\]
Then
\begin{equation*}%\label{eq:hphaserelations}
F'(q)=a\bigl(C-\mathscr H(q)\bigr), \qquad
F''(q)=-a\mathscr H'(q),
\end{equation*}
and every critical point $q$ satisfies
\begin{equation*}%\label{eq:hcritvalue}
F(q)=a\left(\frac q{\sin q}\right)^2.
\end{equation*}

On the affine complexification $H^1_{\C}=\C^3$ of the Heisenberg group, let
\[
\mathcal H_{\C}=\frac12(h_1^2+h_2^2),\qquad h_1=p_1+x_2p_y,\qquad h_2=p_2-x_1p_y.
\]
Complex normal geodesics are projections of the Hamiltonian flow with complex initial covectors; see \cite{BGG}. For a unit-time trajectory, write $p_y=\lambda$ and $v=(h_1(0),h_2(0))$. The Hamilton equations give
\begin{equation}\label{eq:hendpoint}
\begin{aligned}
x(s)&=\frac{\sin(\lambda s)}{\lambda}R_{-\lambda s}v,\qquad\quad
y(1)=\frac{v\cdot v}{2\lambda^2}
\bigl(\lambda-\sin\lambda\cos\lambda\bigr),\\
\mathcal A(\gamma)&=\int_0^1\mathcal H_{\C}\,\dd s=\frac12v\cdot v.
\end{aligned}
\end{equation}

\begin{theorem}\label{thm:hactionclass}
\begin{itemize}
    \item[$(1)$] If $\lambda\notin\pi\Z\setminus\{0\}$, then
\begin{equation}\label{eq:hvfull}
 v=\frac{\lambda}{\sin\lambda}R_\lambda x,
\end{equation}
and
\begin{equation}\label{eq:hcomplexstationarity}
 y=\frac{\sigma}{2}\mathscr H(\lambda),
 \qquad
 \mathcal A
 =\frac{\sigma\lambda^2}{2\sin^2\lambda}
 =F(\lambda).
\end{equation}
Equivalently, the endpoint equation is $F'(\lambda)=0$, and $\mathcal A$ is the critical value.
    \item[$(2)$] If $\lambda=j\pi\ne0$, a finite trajectory requires $x=0$ and
\begin{equation}\label{eq:hquadric}
 v\cdot v=2j\pi y,\qquad
 \mathcal A=j\pi y.
\end{equation} 
\end{itemize}
\noindent Thus, away from the nonzero lattice $\pi\Z\setminus\{0\}$, finite complex normal actions are exactly the critical values of $F_{\sigma,y}$, and the endpoint map $E(v,\lambda)$ satisfies
\[
\det DE(v,\lambda)=-\left(\frac{\sin\lambda}{\lambda}\right)^2F''(\lambda).
\]
Thus conjugacy is equivalent to degeneracy of the phase critical point. The lattice actions occur only at vertical endpoints.
\end{theorem}

\begin{proof}
Away from the nonzero lattice, invert the horizontal equation in \eqref{eq:hendpoint}. Substitution into the vertical equation and $\mathcal A=\frac12v\cdot v$ gives \eqref{eq:hcomplexstationarity}.
At $\lambda=j\pi\ne0$, the horizontal endpoint must vanish and the vertical equation reduces to \eqref{eq:hquadric}. Differentiating \eqref{eq:hendpoint} and taking the Schur complement
of
\[
\partial_vx=\frac{\sin\lambda}{\lambda}R_{-\lambda}
\]
gives the determinant formula.
\end{proof}

For real endpoints with $r>0$, the phase and action are related by 
\[
\mathscr H(q)=C\Longleftrightarrow F'(q)=0,
\qquad F(q)=\mathcal A.
\]

\section{Borel continuation and action singularities}\label{sec:hglobal}

For fixed real endpoints, we classify the critical actions and prove inverse continuation without escape. We then use Morse and cubic cycle transport to determine the finite Borel singularities, followed by the near-vertical large-order asymptotics.

\subsection{The local Borel germ at the minimizing saddle}

Fix $r,y>0$ and set $C=y/a>0$. Since
\[
F'(q)=a\bigl(C-\mathscr H(q)\bigr)\quad\text{and}
\quad\mathscr H'(q)=\frac{2(\sin q-q\cos q)}{\sin^3q}>0,\quad 0<q<\pi,
\]
there is a unique $q_0\in(0,\pi)$ satisfying $\mathscr H(q_0)=C$, and
\[
F''(q_0)=-a\mathscr H'(q_0)<0.
\]
Thus $q_0$ is the nondegenerate minimizing critical point, and we set
\[
A_1:=F(q_0).
\]

For $X>A_1$ sufficiently close to $A_1$, the equation $F(q)=X$ has two conjugate solutions
\[
q_{\mathrm{up}}(X),\qquad
q_{\mathrm{down}}(X)=\overline{q_{\mathrm{up}}(X)}, \qquad \im q_{\mathrm{up}}(X)>0.
\]
With $W(q)=q/\sin q$, the pushforward of the oriented Fourier contour to the action variable $X$ has local density
\begin{equation}\label{eq:hlocaldensity}
\rho(X-A_1)=iW(q_{\mathrm{down}}(X))\frac{\dd q_{\mathrm{down}}}{\dd X}(X)-iW(q_{\mathrm{up}}(X))\frac{\dd q_{\mathrm{up}}}{\dd X}(X).
\end{equation}
The Fourier-contour orientation fixes the signs.

Set $\xi=X-A_1$. The holomorphic Morse lemma gives
\[
q_{\mathrm{up/down}}(A_1+\xi)
=q_0\pm i\sqrt{\frac{2\xi}{-F''(q_0)}}+O(\xi),
\]
and hence
\[
\rho(\xi)=\frac{\sqrt2\,W(q_0)}{\sqrt{-F''(q_0)}}\,\xi^{-1/2}+O(\xi^{1/2}).
\]
Thus $\rho(\xi)=\sum_{n\ge0}r_n\xi^{n-1/2}$ is a convergent Puiseux germ. To obtain the heat-kernel asymptotics from this germ, deform the real $\tau$-contour to $\R+iq_0$. Since $0<q_0<\pi$, no pole is crossed. Writing $\tau=x+iq_0$, so that $q=q_0-ix$, we have
\[
\re F(q_0-ix)-A_1=\frac{2a\sinh^2x\,(x\coth x-q_0\cot q_0)}{\cosh 2x-\cos 2q_0}>0,
\qquad x\ne0.
\]
Hence $iq_0$ is the unique minimizing saddle on the deformed contour, and the contribution away from a fixed neighborhood of it is exponentially smaller.

Termwise Laplace integration therefore gives
\begin{equation}\label{eq:hlocalasymptotic}
p_t(r,y)\sim\frac{e^{-A_1/t}}{4\pi^2t^{3/2}}\,S(t),\qquad 
S(t):=\sum_{n\geq0}r_n\Gamma\!\left(n+\frac12\right)t^n.
\end{equation}
Convergence of the Puiseux series makes $S(t)$ 1-Gevrey.

\begin{definition}\label{def:hresurgence}
A series $S(t)=\sum c_nt^n$ is called 1-Gevrey if $|c_n|\le MA^{n}n!$ for some $A,M>0$. The 1-Gevrey series $S$ is called resurgent if $\widehat{S}$ admits analytic continuation along every finite piecewise smooth path avoiding a locally finite singular set. See \cite{DelabaerePham,KS2022,KS2024} for more details.
\end{definition}

\subsection{Critical action spectrum}
We now classify the roots of the critical equation for $C=y/a>0$:
\[\mathscr H(q)=C.\]

\begin{lemma}\label{lem:hrealcritical}
For $k\ge1$, let $\theta_k$ be the unique root of $\tan\theta_k=\theta_k$ in $(k\pi,k\pi+\pi/2)$. On $(k\pi,(k+1)\pi)$, there are two simple roots $q_{k,+}<\theta_k<q_{k+1,-}$ if $C>\theta_k$, one double root if $C=\theta_k$, and no root if $C<\theta_k$.
\end{lemma}

\begin{proof}
On $(k\pi,(k+1)\pi)$, $\mathscr H$ tends to $+\infty$ at both ends and has its unique minimum at $\theta_k$, with $\mathscr H(\theta_k)=\theta_k$. The three cases follow.
\end{proof}

\begin{corollary}\label{thm:hclassification}
Assume $C\ne\theta_k$ for all $k$, and set $m=\#\{k\ge1:\theta_k<C\}$. Then the positive critical points of $F$ are
\[
q_{1,-}:=q_0,\qquad q_{k,+}\ (1\le k\le m),\qquad q_{k+1,-}\ (1\le k\le m),
\]
so their number is $1+2m$, with
\begin{equation}\label{eq:hcount}
\frac C\pi-\frac32<m<\frac C\pi,\qquad
m=\frac{2y}{\pi r^2}+O(1).
\end{equation}
Writing $A_k=F(q_{k,-})$, $B_k=F(q_{k,+})$, we have
\begin{equation*}%\label{eq:horder}
A_1<B_1<A_2<\cdots<A_m<B_m<A_{m+1},
\end{equation*}
and
\begin{equation*}%\label{eq:hpoleorder}
A_k<k\pi y<B_k,\qquad 1\le k\le m.
\end{equation*}
\end{corollary}

\begin{proof}
The bounds $k\pi<\theta_k<(k+\tfrac12)\pi$ give \eqref{eq:hcount}. At critical points, $F(q)=a(q/\sin q)^2$, and on each interval $(k\pi,(k+1)\pi)$, the value $B_k$ is the local minimum while $A_{k+1}$ is the local maximum, hence $B_k<A_{k+1}$. Finally, writing $q=k\pi+\delta$, $-\pi<\delta<\pi$, we have
\[
F(q)-k\pi y
=a\frac{k\pi(\delta+\sin\delta\cos\delta)+\delta^2}{\sin^2\delta}.
\]
The numerator has the sign of $\delta$ at the two critical roots, which yields $A_k<k\pi y<B_k$. 
\end{proof}

At resonance $C=\theta_j$, since
\begin{equation*}
F'''(\theta_j)=-a\mathscr H''(\theta_j)
=-\frac{2a\theta_j}{\sin^2\theta_j}\ne0,
\end{equation*}
there are $2j$ distinct critical points and
$2j+1$ counted with multiplicity. The inequalities remain strict except for the two values that merge at the cubic point.

\begin{proposition}\label{prop:hrealaction}
If $F'(z)=0$ and $F(z)>0$, then $z$ is real.
\end{proposition}

\begin{proof}
If $F'(z)=0$ and $F(z)>0$, then $W_*:=z/\sin z\in\R\setminus\{0\}$. The critical equation gives
\[
\sin z=z/W_*,\qquad \cos z=W_*-Cz/W_*.
\]
Writing $z=u+iv$ gives
\begin{align*}
\sin u\cosh v&=u/W_*,&\cos u\sinh v&=v/W_*,\\
\sin u\sinh v&=Cv/W_*,&\cos u\cosh v&=W_*-Cu/W_*.
\end{align*}
If $v\ne0$, these identities imply
\[
u=Cv\coth v,\qquad W_*^2=(1+C^2)v\coth v.
\]
Applying $\sin^2u+\cos^2u=1$ to the second and third equations gives
\[W_*^2=(1+C^2)v^2/\sinh^2v.\] 
Thus $v=\sinh v\cosh v$, impossible for real nonzero $v$. Hence $v=0$.
\end{proof}

For the complex critical points, write
\begin{equation*}
F'(z)=\frac{aG_C(z)}{\sin^2z},\qquad
G_C(z)=C\sin^2z+\sin z\cos z-z.
\end{equation*}

\begin{lemma}\label{lem:hremotecritical}
For fixed $a,y>0$, the upper-half-plane nonreal zeros $z_n$ of $F'$ are locally finite and satisfy
\begin{align}
\re z_n&=n\pi+O_{a,y}(1),\notag\qquad
\im z_n=\tfrac12\log|n|+O_{a,y}(1),\notag\\[1mm]
\im F(z_n)&=-an\pi+\tfrac y2\log|n|+O_{a,y}(1).
 \label{eq:hremotecriticalvalue}
\end{align}
The estimates are uniform on compact subsets of $(0,\infty)^2$. The lower-half-plane roots give the other two complex-conjugate tails. Consequently every horizontal action strip contains only finitely many critical values of $F$. For fixed $a,y>0$, there is $\epsilon>0$ such that
\[
\{\re X>0,\ |\im X|<\epsilon\}
\]
contains no nonreal critical value. 
\end{lemma}

\begin{proof}
Set
\[b_C=\frac{-C+i}{4},\quad c_C=-\frac{C+i}{4},\]
Then
\[G_C(z)=C\sin^2z+\sin z\cos z-z
=\frac C2-z+c_Ce^{2iz}+b_Ce^{-2iz}.
\]
For an upper-half-plane zero with large modulus, the dominant balance is
$b_Ce^{-2iz}\sim z$. Hence
\[
\im z=\frac12\log|z/b_C|+O(|z|^{-1}),\qquad
\re z=n\pi+\beta_{\operatorname{sgn}n}
+O\!\left(\frac{\log|n|}{|n|}\right),
\]
where
\[\beta_+=\frac{1}{2}\arg b_C\in(\pi/4,\pi/2),\qquad \beta_-=\frac{1}{2}(\arg b_C-\pi).\]
Conversely, let $z_n^{(0)}$ denote the approximate root above and set
\[
\mathbb D_n=B\!\left(z_n^{(0)},\,M\frac{\log|n|}{|n|}\right).
\]
Writing $z=z_n^{(0)}+w$, the critical equation becomes $w=\Phi_n(w)$, with
\[
\sup_{\mathbb D_n}|\Phi_n'(w)|\le\kappa<1
\]
for $|n|\gg1$, uniformly on compact positive parameter sets. Hence $\mathbb D_n$ contains a unique simple zero, denoted by $z_n$, and
\[
z_n-z_n^{(0)}
=O\!\left(\frac{\log|n|}{|n|}\right).
\]
Since $\cot z_n=-i+O(e^{-2\im z_n})$,
\[
F(z_n)=(y-ia)z_n+O_{a,y}(1),
\]
which yields \eqref{eq:hremotecriticalvalue}. The remaining critical points lie in a fixed compact set. Proposition~\ref{prop:hrealaction} then gives the critical-value-free lateral band.
\end{proof}

\begin{proposition}\label{prop:hactiondistinct}
Let $a,y>0$ and $C=y/a$. Then:
\begin{itemize}[partopsep=0pt,itemsep=0pt]
    \item[$(1)$] Distinct critical points of $F(q)$ have distinct critical values. 
    \item[$(2)$] Every real critical value is positive and comes from a real critical point. 
    \item[$(3)$] The only multiple critical points are the real cubic resonances described above.
\end{itemize}
\end{proposition}
\begin{proof}
At a critical point $q$, set $d=F(q)/a$ and $W=q/\sin q$. Then $d=W^2\ne0$ and
\begin{equation}\label{eq:hcriticalaffinetrig}
 \sin q=\frac qW,\qquad \cos q=W-\frac{Cq}{W}.
\end{equation}
The trigonometric identity gives
\begin{equation*}%\label{eq:hcriticalactionquadratic}
 (1+C^2)q^2-2Cdq+d^2-d=0.
\end{equation*}
Suppose two distinct critical points $q_1,q_2$ give the same $d$. Then
\[
q_1+q_2=\frac{2Cd}{1+C^2},\qquad
q_1q_2=\frac{d^2-d}{1+C^2}.
\]
Writing $W_2=\epsilon W_1$ with $\epsilon\in\{1,-1\}$ and using \eqref{eq:hcriticalaffinetrig}, we obtain
\[e^{i(q_1+q_2)}=\epsilon\left(d+(i-C)(q_1+q_2)+\frac{(i-C)^2}{d}q_1q_2\right)=-\epsilon\frac{(i-C)^2}{1+C^2}.\]
The last expression has modulus one. Hence $\im(q_1+q_2)=0$; since $C>0$, it follows that $d\in\R$.

If $d<0$, write $W=iK$ with $K\in\R\setminus\{0\}$ and $q=u+iv$. Then \eqref{eq:hcriticalaffinetrig} gives
\begin{align*}
K\sin u\cosh v&=v,\qquad& K\cos u\sinh v&=-u,\\
K\cos u\cosh v&=-Cv,&\qquad
-K\sin u\sinh v&=K^2+Cu.
\end{align*}
It follows that $u=Cv\tanh v$ and $K^2=-(1+C^2)v\tanh v<0$, a contradiction. Therefore $d>0$, and Proposition \ref{prop:hrealaction} forces the critical point to be real. The ordering of the real critical values then excludes two distinct points with the same value.

A multiple critical point satisfies $\mathscr H'(q)=0$, hence $\sin q=q\cos q$; substituting in $\mathscr H(q)=C$ gives $q=C$. Thus it is real and obeys $\tan C=C$. The nonzero third derivative of $F$ shows that these points are precisely the cubic resonances.
\end{proof}

We next estimate the distribution and separation of the critical values.

For the fixed real endpoint, let
\[
\Sigma_F:=\CritVal(F)=\Sigma_H^{\rm fin}(r,y)
\]
denote the set of finite critical values. 

By Lemma \ref{lem:hremotecritical}, the upper-half-plane critical roots satisfy
\[
z_n=n\pi+\beta_{\operatorname{sgn}n}
       +\frac i2\log\frac{|n|\pi}{|b_C|}
       +O\!\left(\frac{\log|n|}{|n|}\right),\qquad |n|\longrightarrow\infty.
\]
The constants below may depend on the fixed endpoint.

\begin{lemma}\label{lem:tails}
There are constants $\gamma_+,\gamma_-$ such that the critical values of the upper-half-plane roots $D_n=F(z_n)$ satisfy
\begin{equation}\label{eq:actiontails}
D_n=(y-ia)\left(n\pi+\frac i2\log|n|+\gamma_{\operatorname{sgn}n}\right)+O\!\left(\frac{\log|n|}{|n|}\right).
\end{equation}
The lower-half-plane roots give the other two complex-conjugate tails.
\end{lemma}

\begin{proof}
Set $u=e^{2iz_n}$. The critical equation becomes
\[
b_Cu^{-1}=z_n-C/2-c_Cu,\qquad c_C=-(C+i)/4.
\]
Hence $u=b_C/z_n+O(z_n^{-2})$. Since $\cot z_n=-i(1+2u+O(u^2))$,
\[
F(z_n)=(y-ia)z_n-2iab_C+O(z_n^{-1})=(y-ia)(z_n+i/2)+O(z_n^{-1}).
\]
Substitution of the root expansion proves \eqref{eq:actiontails}, with $\gamma_\pm=\beta_\pm+\frac i2(\log(\pi/|b_C|)+1)$. Conjugation applies because $F$ has real parameters.
\end{proof}

\begin{proposition}
\label{prop:spectrum}
For a fixed real endpoint, there exist $\delta_*>0$ and $K_0,K_1<\infty$ such that
\begin{align}
 |A-B|&\ge\delta_*\qquad(A,B\in\Sigma_F,\ A\ne B),\label{eq:separation}\\
 \#\bigl(\Sigma_F\cap\overline{B(z,R)}\bigr)&\le K_0+K_1R
 \qquad(z\in\C,\ R\ge0).\label{eq:count}
\end{align}
Every affine straight line in the action plane meets $\Sigma_F$ in a finite set.
\end{proposition}

\begin{proof}
Set $v=y-ia$. Along each tail, \eqref{eq:actiontails} gives
\[
\re\frac{D_n}{v}=n\pi+\re\gamma_{\operatorname{sgn}n}+o(1),
\qquad
\text{and hence }~
\re\frac{D_{n+1}-D_n}{v}=\pi+o(1).
\]
Thus consecutive projected points are separated by at least $\pi/2$ for
$|n|\gg1$. Since the projection
\[
\pi_v(D):=\re(D/v)
\]
maps a disk of radius $R$ to an interval of length $O(R)$, each tail contains only $O(R)$ points in such a disk. Summing the four tails and the finite central part gives \eqref{eq:count}.

The four asymptotic ray directions are
\[
\pm v,\qquad \pm\bar v,
\]
and are distinct because $a,y>0$. Hence, for points on two different tails,
\[
|D_n-D_m|\ge c\bigl(|n|+|m|\bigr)-O(\log|n|+\log|m|)
\]
for some $c>0$. Together with the estimate on each tail, this gives uniform separation outside a compact set. Distinctness of the remaining finitely many critical values gives \eqref{eq:separation}.

For a line with direction $e^{i\theta}$, the transverse projection of one tail is
\[
\im(e^{-i\theta}D_n)=\pi n\,\im(e^{-i\theta}(y-ia))
       +\tfrac12\log|n|\,\re(e^{-i\theta}(y-ia))+O(1).
\]
The transverse projection leaves every bounded interval: either the linear coefficient is nonzero, or the logarithmic coefficient is. Applying this to all four tails gives the finite-intersection claim.
\end{proof}

\begin{remark}
Although $\Sigma_F$ is uniformly separated, the difference set $\Sigma_F-\Sigma_F$ is not locally finite: by \eqref{eq:actiontails},
\[
D_{n+1}-D_n\longrightarrow \pi(y-ia)
\]
through infinitely many distinct values. Thus transition frequencies may accumulate, while $\Sigma_F-A$ remains locally finite for every fixed source action $A$. 
\end{remark}

\subsection{Inverse continuation and no-escape}
We study the inverse equation
\[F(z)=X\]
over bounded action paths and rule out escape to infinity or to a pole of $F$.

To remove the poles, define
\begin{equation*}%\label{eq:hclearedinverse}
\mathcal R_X(z):=y\sin z-X\frac{\sin z}{z}+a\cos z
=\frac{\sin z}{z}\bigl(F(z)-X\bigr),
\end{equation*}
where the right-hand side is understood by holomorphic continuation at $z=0$. Since
\[
\mathcal R_X(k\pi)=a(-1)^k\ne0,
\qquad k\in\Z\setminus\{0\},
\]
the zeros of $\mathcal R_X$ are exactly the solutions of $F(z)=X$ in the pole-free domain, with the same multiplicities. Set
\[
\nu:=\sqrt{a^2+y^2},\qquad\alpha:=\arctan(a/y),
\qquad\lambda_n:=n\pi-\alpha.
\]
Then
\[
y=\nu\cos\alpha,\qquad a=\nu\sin\alpha,
\]
and at $X=0$,
\begin{equation}\label{eq:hR0}
\mathcal R_0(z)=y\sin z+a\cos z=\nu\sin(z+\alpha),
\end{equation}
so its zeros are precisely the points $\lambda_n$.

\begin{lemma}
\label{lem:hcomplexnoescape}
Let $\mathcal K\Subset\C$ and $\mathcal P\Subset(0,\infty)^2$. There exist $\rho>0$, $N<\infty$, $R<\infty$, uniform for $X$ in a neighborhood of $\mathcal K$ and
$(a,y)\in\mathcal P$, such that every zero of $\mathcal R_X$ with $|z|>R$ is the unique simple zero in one of the disks $\mathbb D(\lambda_n,\rho)$, $|n|\ge N$, and
\begin{equation}\label{eq:hinverseroot}
z_n(X)=\lambda_n-\frac{aX}{a^2+y^2}\frac1{\lambda_n}
+O(|n|^{-2}),\qquad |n|\to\infty,
\end{equation}
Consequently, every inverse lift of $F(z)=X$ over a compact noncritical action path remains relatively compact in the pole-free
domain and therefore continues to the end of the path.
\end{lemma}

\begin{proof}
Choose $M>0$ so that the relevant neighborhood of $\mathcal K$ is contained in $|X|\le M$. Fix $0<\rho<\pi/4$. By
\eqref{eq:hR0} and periodicity,
\[
|\sin(z+\alpha)|\ge c_\rho e^{|\im z|}
\]
outside the disks $\mathbb D(\lambda_n,\rho)$, while
\[
|\mathcal R_X(z)-\mathcal R_0(z)|
=\left|X\frac{\sin z}{z}\right|
\le \frac{C M}{|z|}e^{|\im z|}.
\]
Since $\nu$ is bounded away from zero on $\mathcal P$, Rouch\'e's theorem applies uniformly for $|z|$ sufficiently large. Hence each disk $\mathbb D(\lambda_n,\rho)$ with $|n|\ge N$ contains exactly one simple zero of $\mathcal R_X$. 

Write $z_n=\lambda_n+\delta_n$. The inverse equation first gives $\delta_n=O(|n|^{-1})$. Using $\sin(\lambda_n+\alpha)=0$ and substituting into $\mathcal R_X(z_n)=0$ gives
\[
\delta_n=-\frac{aX}{a^2+y^2}\frac1{\lambda_n}
+O(|n|^{-2}).
\]

Now let $X(s)$ be a compact path avoiding $\Sigma_F$, and let $z(s)$ be an inverse lift. By the uniform separation above, a lift in $\mathbb D(\lambda_n,\rho)$, $|n|\ge N$, remains in the same disk, whereas every other lift remains in a fixed bounded region. Thus no lift can escape to infinity.

A lift also cannot approach a pole $k\pi\ne0$, since
\[
F(z)=\frac{ak\pi}{z-k\pi}+O(1)
\qquad (z\to k\pi),
\]
whereas $X(s)$ remains bounded. Hence the lifted path is relatively compact in the pole-free domain.

If a maximal lift were defined only on $[0,s_*)$, compactness would give a limiting point $z_*$ with
\[
F(z_*)=X(s_*).
\]
Since $X(s_*)\notin\Sigma_F$, we have $F'(z_*)\ne0$, and the inverse-function theorem extends the lift past $s_*$, a contradiction. Thus the lift continues along the whole path. If the prescribed endpoint is critical, the final continuation is supplied by the corresponding finite Puiseux chart.
\end{proof}

\begin{proposition}
\label{prop:hregularcovering}
Let $U=\C\setminus(\pi\Z\setminus\{0\})$ and $\Sigma_F=\CritVal(F)$. Then
\[
F:U\setminus F^{-1}(\Sigma_F)\longrightarrow
\C\setminus\Sigma_F
\]
is a connected covering with countably infinite fibers. If $C\ne\theta_j$, $j\ge1$, then its monodromy group on every regular fiber is
\[
\Sym_{\rm fin}\bigl(F^{-1}(X_0)\bigr).
\]
In particular, every ordered difference of two distinct inverse sheets is a continuation of the initial two-sheet cycle.
\end{proposition}

\begin{proof}
Lemma \ref{lem:hcomplexnoescape} gives, near every regular value, one holomorphic inverse in each disk $\mathbb D(\lambda_n,\rho)$, $|n|\ge N$, and only finitely many additional inverse branches in a bounded region. After shrinking the action disk, these charts are disjoint; hence
\[
F:U\setminus F^{-1}(\Sigma_F)\longrightarrow\C\setminus\Sigma_F
\]
is a covering with countably infinite fibers.

Since $F^{-1}(\Sigma_F)$ is locally finite in the connected pole-free surface $U$, its complement is path connected. Thus the covering is connected and its monodromy acts transitively.

Assume now that $C\ne\theta_j$. By Proposition \ref{prop:hactiondistinct}, each critical value has a unique Morse critical point above it, so a meridian acts by a transposition. These transpositions act transitively on the fiber; hence they generate the full finitely supported symmetric group $\Sym_{\rm fin}\bigl(F^{-1}(X_0)\bigr)$.
\end{proof}

The following continuation result also applies to the horizontal case.

\begin{proposition}
\label{prop:hcontinuationAbel}
Let $g:U\to\C$ be holomorphic on a connected pole-free Riemann surface, and let $\varpi$ be holomorphic. Assume that  $\Sigma_g:=\CritVal(g)$ is locally finite and that inverse lifts over compact subsets of $\C\setminus\Sigma_g$ do not escape. Then
\begin{itemize}
\setlength{\itemsep}{-1pt}
    \item[$(1)$] every finite inverse-pushforward density, and its Abel transform, admits analytic continuation along every finite path in $\C\setminus\Sigma_g$.
    \item[$(2)$] every Morse critical point $z_c$ is accessible by some continuation of a single inverse branch. If $\varpi(z_c)\ne0$, the continued density has a nonzero local term
\[
c\,(u-g(z_c))^{-1/2},\qquad c\ne0.
\]
\end{itemize}
Consequently, the presence of any finite nonzero singularity implies exact Gevrey order one and zero convergence radius for the associated formal series.
\end{proposition}

\begin{proof}
The no-escape hypothesis and the inverse-function theorem continue each inverse branch along every finite path in $\C\setminus\Sigma_g$. The same holds for finite pushforwards and their Abel transforms.

Since $g^{-1}(\Sigma_g)$ is locally finite, its complement is path connected, so any Morse critical point is accessible from a regular inverse branch. In a Morse coordinate
\[
g(z)-g(z_c)=w^2,
\]
the pushforward of $\varpi(z)\,dz$ has a nonzero $(u-g(z_c))^{-1/2}$ term whenever $\varpi(z_c)\ne0$. The final Gevrey statement follows from Cauchy--Hadamard.
\end{proof}

\begin{corollary}
\label{prop:hlinearresurgence}
Let $q_*$ be a Morse critical point of $F$ and set $A_*=F(q_*)$. The corresponding pushed density $\rho_*$ and its ordinary Borel transform admit analytic continuation along every finite path avoiding
\[
(\Sigma_F-A_*)\cup\{0\},
\]
and have no other possible nonzero finite singularities.
\end{corollary}

\begin{remark}\label{rmk:accessibility}
Every finite critical point of $F$ is accessible from either initial inverse sheet by analytic continuation along a finite action path avoiding $\Sigma_F$ in its interior. Joining the chosen initial inverse point to a regular point near a prescribed critical point and projecting by $F$ gives the required continuation path.
\end{remark}

The Abel transform determines the local ordinary Borel singularity. For the saddle series $S(t)$ in \eqref{eq:hlocalasymptotic}, write
\[
S(t)=\sum_{n\ge0}c_nt^n, \qquad
c_n=r_n\Gamma\!\left(n+\frac12\right),
\]
so the half-order Borel transform is the pushed density:
\[
\mathcal B_{1/2}S(\xi)=\sum_{n\ge0}\frac{c_n}{\Gamma(n+1/2)}\xi^{n-1/2}=\rho(\xi).
\]
The ordinary Borel transform is the Abel transform of $\rho$:
\begin{equation}\label{eq:hAbelconvention}
\widehat S(\xi)=\frac1{\sqrt\pi}
\int_0^\xi\frac{\rho(u)}{\sqrt{\xi-u}}\dd u.
\end{equation}
The square roots and contour are transported from the initial positive segment.

\begin{lemma}
\label{lem:hAbelsingular}
\begin{itemize}
\setlength{\itemsep}{-1pt}
    \item[$(1)$] If near a nonzero
singularity $\omega$, $\rho(u)=c(\omega-u)^{-1/2}+\text{less singular terms}$, $c\ne0$, then
\[
\widehat S(\xi)=-\frac{c}{\sqrt\pi}\log(\omega-\xi)
 +\text{less singular terms}.
\]
\item[$(2)$] If instead $\rho(u)=c(\omega-u)^{-2/3}+\text{less singular terms}$, $c\ne0$, then
\[\widehat S(\xi)=
\frac{c}{\sqrt\pi}\mathrm B\!\left(\frac12,\frac16\right)
(\omega-\xi)^{-1/6}+\text{less singular terms}.
\]
\end{itemize}
\end{lemma}

\begin{proof}
Only the terminal part of the Abel contour is singular. With $\delta=\omega-\xi$, the leading term is
\[
\frac c{\sqrt\pi}\int_\delta^\eta v^{-\alpha}(v-\delta)^{-1/2}\dd v.
\]
For $\alpha=1/2$, this equals $-c\pi^{-1/2}\log\delta+O(1)$, while for $\alpha=2/3$, the scaling $v=\delta s$ gives
\[c\pi^{-1/2}\mathrm B(1/2,1/6)\delta^{-1/6}+O(1).\]
All earlier contour pieces are holomorphic, so neither leading term can cancel.
\end{proof}

\subsection{Transport of inverse branches}
\subsubsection{Morse transport}
Fix a regular value $X\notin\Sigma_F$ and write $\mathcal F_X:=F^{-1}(X)$ in the pole-free $q$-plane. A reduced zero-cycle on $\mathcal F_X$ is a finitely supported formal sum
\[
\delta=\sum_j n_j[q_j(X)],\qquad
n_j\in\Z,\qquad \sum_j n_j=0,
\]
where $[q_j(X)]$ denotes the point $q_j(X)\in\mathcal F_X$ regarded as a generator of $\widetilde H_0(\mathcal F_X;\Z)$.

Let $D$ be a Morse critical value at which the inverse branches $q_i$ and $q_j$ collide. Analytic continuation around a small meridian of $D$ exchanges these two sheets; denote the induced monodromy on zero-cycles by $T_{ij}$. Then
\[
(T_{ij}-I)\delta=(n_j-n_i)\bigl([q_i]-[q_j]\bigr),
\]
which is the usual one-dimensional Picard--Lefschetz variation \cite{DelabaereHowls,Pham1983}. Its scalar pushforward is
\begin{equation}\label{eq:hzerocycledensity}
\mathcal Q_\delta(X)=i\sum_j n_j\frac{W(q_j(X))}{F'(q_j(X))}.
\end{equation}

Along the positive real action axis, the real critical-point classification and Proposition~\ref{prop:hrealaction} give, for nonresonant $C>0$,
\begin{equation*}
\begin{array}{c|c}
 A_k<X<B_k\ (1\le k\le m)&\text{one nonreal conjugate pair},\\
 B_k<X<A_{k+1}\ (1\le k\le m)&\text{all roots real},\\
 X>A_{m+1}&\text{one nonreal conjugate pair}.
\end{array}
\end{equation*}
For $m=0$ only the last line occurs. The pair is created at a maximum $A_k$ and becomes two real roots at a minimum $B_k$. To fix the orientations, choose increasing Morse coordinates near a maximum and the following minimum:
\[
F=A_k-u_A^2, \qquad F=B_k+u_B^2.
\]
Upper and lower detours at $A_k$ send the increasing real branch into the upper and lower $q$ half-planes, respectively. Continuing the same square-root branch through $B_k$ identifies the outgoing increasing real branch. 

\begin{proposition}
\label{prop:hlabelinduction}
Assume that $C$ is nonresonant and label the roots over $X=0$ by
\[
q_n(0)=n\pi-\alpha, \qquad \alpha=\arctan(a/y).
\]
Along $\Gamma_+$, the two sheets colliding at $A_k$ have labels $(0,k)$. After the subsequent minimum $B_k$, label $0$ continues along the increasing real branch toward $A_{k+1}$, while label $k$ follows the pole-side branch and tends to $k\pi$ as $X\to\infty$. Along $\Gamma_-$ the nonreal pair is conjugated, but the same label transport holds. After $A_{m+1}$, the continued pair is the unique nonreal pair for $X>A_{m+1}$.
\end{proposition}

\begin{proof}
At $X=0$, labels $0$ and $1$ are the two real branches meeting at $A_1$. Suppose inductively that $(0,k)$ meet at $A_k$. The Morse passage carries them to the unique nonreal pair on $(A_k,B_k)$; at $B_k$, label $0$ exits on the increasing branch, while label $k$ exits on the pole-side branch and tends to $k\pi$.

The increasing branch continues to $A_{k+1}$, where it meets label $k+1$, completing the induction. Conjugation gives the lower lateral statement; uniqueness of the nonreal pair beyond $A_{m+1}$ gives the last claim.
\end{proof}

\subsubsection{Cubic transport at resonance}
\label{sec:hcubicclosure}
At $C=\theta_j$, the local inverse has three branches. We compute their cubic transport and the nonzero contribution of the physical cycle on both lateral continuations. Set
\[
\mathscr D_j=F(\theta_j)=y\theta_j+a,\qquad
\kappa_j=\frac{a\theta_j}{3\sin^2\theta_j}>0.
\]
Then the local phase and inverse density with nonzero amplitude satisfy
\begin{align}
 F(q)-\mathscr D_j
   &=-\kappa_j(q-\theta_j)^3+O((q-\theta_j)^4),\notag\\
 \mathcal Q(X)
   &=C_j^{\rm lat}(X-\mathscr D_j)^{-2/3}
     +O((X-\mathscr D_j)^{-1/3}),
 \qquad C_j^{\rm lat}\ne0.
 \label{eq:hcubicQ}
\end{align}

\begin{lemma}\label{lem:hcubicnormal}
Fix $a>0$, $j\ge1$, and write $C=\theta_j+\mu$. Near $(q,\mu)=(\theta_j,0)$, there are analytic coordinates $u=u(q,\mu)$ and analytic functions $F_c,\Lambda$ such that
\[
F(q;C)=F_c(\mu)+\Lambda(\mu)u-\frac{u^3}{3},
\]
with
\[
u_q(\theta_j,0)=\gamma_j=\left(\frac{a\theta_j}{\sin^2\theta_j}\right)^{1/3},\qquad
\Lambda'(0)=\frac{a}{\gamma_j}>0.
\]
\end{lemma}

\begin{proof}
At $(\theta_j,0)$, we have
\[
F_q=F_{qq}=0,\qquad F_{qqq}=-\frac{2a\theta_j}{\sin^2\theta_j}\ne0,
\qquad\partial_C F_q=a.
\]
The standard parametric $A_2$ normal form gives
\[
F(q;C)=F_c(\mu)+\Lambda(\mu)u-\frac{u^3}{3}.
\]
Comparing the $q$- and $\mu$-derivatives at $(\theta_j,0)$ gives $u_q(\theta_j,0)$ and $\Lambda'(0)$.
\end{proof}

In this chart, the transformed differential
\[
i\frac{q(u,\mu)}{\sin q(u,\mu)}q_u(u,\mu)\dd u
\]
is holomorphic and nonzero near the cubic point. Thus the inverse map determines the local singularity. Choose regular values
\[
X_\mp=F_c(\mu)\mp\eta
\]
in a fixed local action disk containing the two nearby critical values. Let $\mathfrak p_{\rm up}$ and $\mathfrak p_{\rm down}$ be the upper and lower semicircular paths from $X_-$ to $X_+$ in this disk, and denote by
\[T_{\rm up},T_{\rm down}:
\widetilde H_0(\mathcal F_{X_-};\Z)
\longrightarrow
\widetilde H_0(\mathcal F_{X_+};\Z).\]
the induced analytic continuation maps on reduced zero-cycles.

At $\mu=0$, let $A,B,C_0$ be the three inverse roots over $X_-$ and $a_0,b,c$ the corresponding roots over $X_+$, ordered counterclockwise.
Use the reduced zero-cycle bases
\begin{equation}\label{eq:hcubicbases}
v_1=B-A,\qquad v_2=A-C_0;\qquad
w_1=b-a_0,\qquad w_2=a_0-c.
\end{equation}
The scalar pushforward \eqref{eq:hzerocycledensity} transfers cycle identities to densities.

\begin{theorem}\label{thm:hcubicblock}
In the bases \eqref{eq:hcubicbases}, the continuation maps $T_{\rm up}$ and $T_{\rm down}$ are represented by
\begin{equation}\label{eq:hcubicP}
P_{\rm up}=
\begin{pmatrix}1&0\\1&-1\end{pmatrix},\qquad
P_{\rm down}=
\begin{pmatrix}-1&1\\0&1\end{pmatrix}.
\end{equation}
These matrices are unchanged for sufficiently small $\mu$. The relative monodromy is
\begin{equation}\label{eq:hcubicM}
M=P_{\rm down}^{-1}P_{\rm up}
=\begin{pmatrix}0&-1\\1&-1\end{pmatrix},
\qquad M^3=I.
\end{equation}
For the positive Fourier orientation, the incoming and outgoing conjugate-pair cycles are
\begin{equation}\label{eq:hcubicphysical}
\Pi_-=C_0-B=-(v_1+v_2),\qquad
\Pi_+=c-b=-(w_1+w_2),
\end{equation}
and
\begin{equation}\label{eq:hcubicphysicaltransport}
T_{\rm up}\Pi_-=-w_1,\qquad
T_{\rm down}\Pi_-=-w_2,\qquad
\Pi_+=T_{\rm up}\Pi_-+T_{\rm down}\Pi_-.
\end{equation}
Thus both cubic components occur with nonzero coefficient.
\end{theorem}

\begin{proof}
At $\mu=0$, analytic continuation along the upper and lower semicircular paths acts on the three local inverse branches by
\[
(A,B,C_0)\longmapsto(c,b,a_0),
\qquad
(A,B,C_0)\longmapsto(b,a_0,c),
\]
respectively. Substitution into \eqref{eq:hcubicbases} gives \eqref{eq:hcubicP}, hence \eqref{eq:hcubicM}. Applying these matrices to $\Pi_-=-(v_1+v_2)$ gives \eqref{eq:hcubicphysicaltransport}. Analytic continuation preserves the permutations for small $\mu$, and \eqref{eq:hzerocycledensity} transfers the cycle identities to the scalar densities.
\end{proof}

\subsection{Positive real action singularities}

In the variable $\xi=X-A_1$, let $\Gamma_\pm$ denote the upper and lower lateral continuations along the positive real action axis. For a Borel germ $H$, write
\[
\operatorname{Sing}_{\Gamma_\pm}H
\]
for the finite singularities reached along these paths.

\begin{theorem}
\label{thm:hpositivespectrum}
Set $G(\xi)=\sqrt{\xi}\rho(\xi)$. For every $r,y>0$,
\[
\operatorname{Sing}_{\Gamma_\pm}G
=\operatorname{Sing}_{\Gamma_\pm}\widehat S
=\bigl(\Sigma_H^{\rm fin}(r,y)-A_1\bigr)\cap(0,\infty).
\]
\begin{enumerate}[label=\textnormal{(\roman*)}]
\item If $C\ne\theta_k$ and $m=\#\{k:\theta_k<C\}$, this set is
\[
\{B_k-A_1:1\le k\le m\}\cup\{A_k-A_1:2\le k\le m+1\}.
\]
\item If $C=\theta_j$, it is
\[
\{B_k-A_1:1\le k<j\}\cup\{A_k-A_1:2\le k\le j\}\cup\{\mathscr D_j-A_1\}.
\]
\end{enumerate}
At each Morse value, the density has a nonzero singular term of exponent $-1/2$, and the ordinary Borel transform has a nonzero logarithmic singularity. At $\mathscr D_j-A_1$ the exponents are $-2/3$ and $-1/6$,
respectively. There are no other finite singularities on these paths.
\end{theorem}

\begin{proof}
Theorem \ref{thm:hactionclass}, Corollary \ref{thm:hclassification}, and Proposition \ref{prop:hrealaction} give exactly the positive action differences, while Lemma \ref{lem:hcomplexnoescape} excludes any finite-action escape or arrival at a pole.

(i) If $C\ne\theta_k$ for all $k\ge1$, Proposition \ref{prop:hlabelinduction} shows that the initial cycle meets every later real Morse value along both lateral paths. At each such collision its vanishing-cycle coefficient is nonzero; since $W(q_c)\ne0$, the corresponding density has a nonzero square-root singularity. Lemma \ref{lem:hAbelsingular} therefore gives a nonzero logarithmic singularity of the ordinary Borel transform.

(ii) If $C=\theta_j$, the earlier Morse transitions remain unchanged, and Theorem \ref{thm:hcubicblock} shows that the cubic contribution is nonzero on both lateral continuations. For $j=1$ the same nonvanishing follows directly from the local cubic expansion. Hence the density has exponent $-2/3$, and Lemma \ref{lem:hAbelsingular} gives the ordinary Borel exponent $-1/6$.

All other inverse points are regular. Lemma \ref{lem:hAbelsingular} gives the ordinary Borel exponents, and Corollary \ref{prop:hlinearresurgence} gives finite-path continuation.
\end{proof}

For a continuation path $\mathfrak p$, $\operatorname{Sing}^*_{\mathfrak p}$ denotes the nonzero finite singularities reached by analytic continuation along $\mathfrak p$.

\begin{theorem}
\label{thm:hfullspectrum}
For every $r,y>0$,
\begin{equation}\label{eq:hfullspectrum}
 \bigcup_{\mathfrak p}\operatorname{Sing}_{\mathfrak p}^* G
 =\bigcup_{\mathfrak p}\operatorname{Sing}_{\mathfrak p}^* \widehat S
 =(\Sigma_F-A_1)\setminus\{0\}.
\end{equation}
Here $\mathfrak p$ ranges over finite continuation paths from the minimizing-saddle germ avoiding the translated critical spectrum. At Morse actions the ordinary Borel singularity is logarithmic; at a cubic action the density and ordinary Borel exponents are $-2/3$ and $-1/6$.
\end{theorem}

\begin{proof}
Corollary \ref{prop:hlinearresurgence} gives
\[
\bigcup_{\mathfrak p}\operatorname{Sing}_{\mathfrak p}^*G\subset(\Sigma_F-A_1)\setminus\{0\},\qquad
\bigcup_{\mathfrak p}\operatorname{Sing}_{\mathfrak p}^*\widehat S\subset(\Sigma_F-A_1)\setminus\{0\}.
\]
Conversely, let $D\in\Sigma_F$, $D\ne A_1$, and let $q_c$ be the unique critical point with $F(q_c)=D$. By the accessibility remark, an initial inverse sheet can be continued to a punctured neighborhood of $q_c$. If $q_c$ is Morse, the continued cycle has a nonzero vanishing-cycle component; if it is cubic, Theorem \ref{thm:hcubicblock} gives a nonzero cubic component. Hence the corresponding density is singular at $\omega=D-A_1$.

Since $\omega\ne0$, multiplication by $\sqrt{\xi}$ does not change the local singularity, and Lemma \ref{lem:hAbelsingular} gives the corresponding singularity of $\widehat S$. Thus every
$\omega\in(\Sigma_F-A_1)\setminus\{0\}$ occurs along some continuation path.
\end{proof}

\subsection{Nearest singularity and large-order growth near the vertical axis}
We now identify the nearest singularity of the initial Borel germ and derive the resulting large-order asymptotics near the vertical axis.

\begin{lemma}
\label{lem:hfirstpairisolation}
Let $K=[y_0,y_1]\Subset(0,\infty)$ and $a=r^2/2$. There exist $r_0,\delta>0$ such that, for $y\in K$ and $0<r<r_0$, the only critical values of $F$ in
\[
|X-\pi y|<\delta
\]
are $A_1$ and $B_1$. Moreover, uniformly for $y\in K$,
\[
q_{1,\mp}=\pi\mp\sqrt{\frac{\pi a}{y}}+O_K(a),\quad 
A_1=\pi y-2\sqrt{\pi ay}+O_K(a), \quad
B_1=\pi y+2\sqrt{\pi ay}+O_K(a).\]
Hence
\begin{equation}\label{eq:homegavertical}
\omega_1:=B_1-A_1
=4\sqrt{\pi ay}+O_K(a)
=2r\sqrt{2\pi y}+O_K(r^2).
\end{equation}
After decreasing $r_0$ if necessary, every other critical value $D$ satisfies $|D-A_1|\ge\delta/2$.
\end{lemma}

\begin{proof}
At $a=0$,
\[\mathcal R_{X}(q)=\sin q\left(y-\frac{X}{q}\right),
\]
so near $X=\pi y$, the only multiple inverse root is
\[
(q,X)=(\pi,\pi y);
\]
all other roots are uniformly simple for $y\in K$. The same Rouch\'e estimate used in Lemma \ref{lem:hcomplexnoescape} is uniform here for $0\le a\le a_0$ and $y\in K$, since $\sqrt{a^2+y^2}\ge y_0>0$. It therefore excludes additional
large-modulus multiple roots. Hence, for small $a$, the only nearby critical values are the two branches $A_1,B_1$ issued from this double root.

Write $q=\pi+s$. Since
\[
\mathscr H(\pi+s)
=\frac{\pi}{s^2}+\frac{\pi}{3}+O(s),
\]
the critical equation $\mathscr H(q)=y/a$ gives
\[
s=\pm\sqrt{\frac{\pi a}{y}}+O_K(a).
\]
At a critical point, $ F(q)=a(q/\sin q)^2$ yields the stated expansions for $A_1,B_1$ and \eqref{eq:homegavertical}. The uniform isolation of the double root at $a=0$ gives the final separation from all other critical values.
\end{proof}

\begin{lemma}
\label{lem:huniformamplitude}
Let $K\Subset(0,\infty)$. For sufficiently small $r>0$, uniformly for $y\in K$, set
\[
\omega=B_1-A_1,\qquad q_+=q_{1,+},\qquad
\chi=-\frac{\sqrt2\,W(q_+)}{\sqrt{F''(q_+)}}>0.
\]
There exist $R>1$ and a function $A_{r,y}$ holomorphic on $|z|<R$ such that
\begin{equation}\label{eq:huniformfactorization}
G(\omega z)=\chi\,A_{r,y}(z)(1-z)^{-1/2},
\qquad A_{r,y}(1)=1.
\end{equation}
Moreover,
\[
A_{r,y}(z)=1+O_K(r) \quad (|z|\le R), \qquad
\chi\sqrt{\omega}=2\pi+O_K(r),
\]
uniformly in the stated parameter region.
\end{lemma}

\begin{proof}
Set $b=\sqrt{\pi a/y}$ and $q=\pi+bu$. By Lemma \ref{lem:hfirstpairisolation},
\[
A_1=\pi y-2yb+O_K(b^2),\qquad\omega=4yb+O_K(b^2).
\]
After the scaling $X=A_1+\omega z$, the inverse equation extends holomorphically to $b=0$ and reduces to
\[
u^2+(2-4z)u+1=0.
\]
Hence, on a fixed disk $|z|<R$ with $R>1$, the two inverse branches form a uniform two-sheeted cover whose only branch points are $z=0$ and $z=1$.

On this cover the pushed differential is $ib\,W(\pi+bu)\dd u$. Therefore
\[
T_{r,y}(z):=\omega\sqrt{z(1-z)}\,\rho(\omega z)
\]
is single-valued and holomorphic on $|z|<R$. The limiting quadratic gives $T_{r,y}(z)=2\pi+O_K(b)$ uniformly, while the local Morse expansion at $q_+$ gives
\[
T_{r,y}(1)=\chi\sqrt{\omega}.
\]
Thus $A_{r,y}(z):=T_{r,y}(z)/T_{r,y}(1)$ satisfies \eqref{eq:huniformfactorization}, $A_{r,y}(1)=1$, and
\[
A_{r,y}=1+O_K(b)=1+O_K(r).
\]
The same estimate at $z=1$ gives $\chi\sqrt{\omega}=2\pi+O_K(r)$.
\end{proof}

\begin{theorem}
\label{thm:hlargeorder}
Under the assumptions and notation of Lemma \ref{lem:huniformamplitude}, define
\[
\beta(r,y):=\frac12 A_{r,y}'(1)-\frac14.
\]
Then $\omega$ is the unique nearest nonzero singularity of the initial germs $G$ and $\widehat S$. Near $\xi=\omega$,
\[
\rho(\xi)=\chi(\omega-\xi)^{-1/2}
\left(1-\frac{2\beta(r,y)}{\omega}(\omega-\xi)+O\bigl((\omega-\xi)^2\bigr)\right),
\]
and
\[
\widehat S(\xi)=-\frac{\chi}{\sqrt\pi}\left(1-\frac{\beta(r,y)}{\omega}(\omega-\xi)+O\bigl((\omega-\xi)^2\bigr)\right)\log\!\left(1-\frac{\xi}{\omega}\right)
+H(\xi),
\]
where $H$ is holomorphic near $\omega$.

If $S(t)=\sum_{n\ge0}c_nt^n$, then
\[
c_n=\frac{\chi}{\sqrt\pi}\Gamma(n)\omega^{-n}
\left(1+\frac{\beta(r,y)}{n}+O_K(n^{-2})\right),
\qquad n\to\infty,
\]
uniformly for $y\in K$ and $0<r<r_0$. Moreover,
\[
\beta(r,y)=-\frac14+O_K(r).
\]
\end{theorem}

\begin{proof}
By \eqref{eq:huniformfactorization}, $\omega$ is the unique nearest nonzero singularity of $G$. The Abel transform gives the same nearest singularity for $\widehat S$. 

Writing $\delta=\omega-\xi$ and expanding at $z=1$ gives
\[\rho(\xi)=\chi\delta^{-1/2}
\left(1+\frac{\frac12-A_{r,y}'(1)}{\omega}\delta+O(\delta^2)
\right),
\]
and termwise Abel integration gives the logarithmic expansion of $\widehat S$.

For the coefficients, set
\[
p_n=[z^n](1-z)^{-1/2}=\frac{\Gamma(n+\frac12)}
{\sqrt\pi\,\Gamma(n+1)}.
\]
A first-order Darboux expansion yields
\[
[z^n]\!\left(A_{r,y}(z)(1-z)^{-1/2}\right)
=p_n\left(1+\frac{A_{r,y}'(1)}{2n}+O_K(n^{-2})\right).
\]
Using $c_n=r_n\Gamma(n+\frac12)$ and
\[
\frac{\Gamma(n+\frac12)^2}{\Gamma(n+1)\Gamma(n)}
=1-\frac1{4n}+O(n^{-2}),
\]
we obtain
\[
c_n=\frac{\chi}{\sqrt\pi}\Gamma(n)\omega^{-n}
\left(1+\frac{\beta(r,y)}{n}+O_K(n^{-2})\right).
\]
Finally, $A_{r,y}=1+O_K(r)$ on a fixed disk about $z=1$, so Cauchy's
estimate gives $A_{r,y}'(1)=O_K(r)$ and hence
\[
\beta(r,y)=-\frac14+O_K(r). \qedhere
\]
\end{proof}

\begin{remark}
The large-order formula uses the exact action difference $\omega=B_1-A_1$. Replacing it by a truncated small-$r$ expansion inside $\omega^{-n}$ is not uniform in $n$.
\end{remark}
Outside the near-vertical regime, Theorem \ref{thm:hfullspectrum} still implies exact Gevrey order one, although the nearest singularity is not identified here.

\section{Global Stokes structure}
\label{sec:hglobalstokes}
Section~\ref{sec:hglobal} determines the Borel singularities. Retaining their action labels, we now construct the directional Stokes operators. 

For each $A\in\Sigma_F$, choose an orientation of the local two-sheet cycle and let $\rho_A$ be its scalar pushforward, written in the local action variable $\xi=F-A$. Thus
\[
\rho_A(\xi)=\sum_{k\ge0}r_{A,k}\xi^{k-1/2},
\qquad
\Phi_A(t)=\sum_{k\ge0}r_{A,k}\Gamma\!\left(k+\frac12\right)t^k,\qquad
\Psi_A(t)=e^{-A/t}\Phi_A(t).
\]

For a direction $\theta$, let $\gamma_{\theta,\pm}$ denote the ray $e^{i\theta}\mathbb R_+$ with upper, respectively lower, semicircular detours around the points of
\[
 (\Sigma_F-A)\cap e^{i\theta}\mathbb R_+,
\]
and let $\rho_A^{\theta,\pm}$ be the corresponding analytic
continuations of $\rho_A$. Define
\[
\mathcal S_{\theta,\pm}\Phi_A(t):=
t^{-1/2}\int_{\gamma_{\theta,\pm}}
e^{-\xi/t}\rho_A^{\theta,\pm}(\xi)\,d\xi,
\qquad
\Psi_A^{\theta,\pm}(t):=
e^{-A/t}\mathcal S_{\theta,\pm}\Phi_A(t).
\tag{4.1}
\]

\subsection{Local connection formulas}

\iffalse
The Morse case is the standard one-dimensional Picard--Lefschetz variation. If the sheets $i_p,j_p$ collide at a Morse point $p$ over $D$, then
\[
(T-I)\delta=\sum_p (n_{j_p}-n_{i_p})\,\delta_{D,p}.
\]
Applying the scalar pushforward and the Abel transform gives the local Borel connection below; Laplace integration gives the corresponding Stokes jump.
\fi

The Morse case is the standard one-dimensional Picard--Lefschetz variation. If the sheets $i_p,j_p$ collide at a Morse point $p$ over $D$, then
\[
(T-I)\delta=\sum_p (n_{j_p}-n_{i_p})\,\delta_{D,p}.
\]
Applying the scalar pushforward and the Abel transform gives the local Borel connection below.

\begin{proposition}
\label{prop:hMorseconnection}
Suppose that all critical points above $D$ are Morse. Let $\delta=\sum_j n_j[q_j]$ be an incoming reduced cycle, and let $\delta_{D,p}$ be the vanishing cycle at $p$, where the sheets $i_p,j_p$ collide.

If $\Phi_\delta$ is based at action $A$ and $\omega=D-A$, then
\[
\Disc_\omega^{\rm loc}\widehat\Phi_\delta(\omega+\eta)=
\sum_p(n_{j_p}-n_{i_p})\,\widehat\Phi_{D,p}(\eta).
\]
%Whenever the lateral Laplace integrals converge,
%\[(\mathcal S_+-\mathcal S_-)\Phi_\delta=e^{-\omega/t}\sum_p(n_{j_p}-n_{i_p})\,\mathcal S\Phi_{D,p}.\]
\end{proposition}

\noindent For the cubic case, apply the scalar pushforward and Abel transform to Theorem \ref{thm:hcubicblock}.

\begin{corollary}
\label{cor:hcubicalien}
In the notation of Theorem~\ref{thm:hcubicblock}, for an incoming coefficient vector $\mathbf c$,
\begin{equation}\label{eq:hcubicBorelblock}
\Disc_{\mathscr D_j}^{\rm loc}\mathcal Q_{\mathbf c}
=\mathcal Q_{(P_{\rm up}-P_{\rm down})\mathbf c}.
\end{equation}
The corresponding ordinary Borel discontinuity is the Abel transform of \eqref{eq:hcubicBorelblock}.

Modulo holomorphic germs, if $M$ is the cubic density monodromy matrix, then the ordinary Borel monodromy is $-M$. Hence the density monodromy has order three, while the ordinary Borel monodromy has order six, with local exponents $\pm1/6$ modulo integers.
\end{corollary}

\begin{proof}
Equation \eqref{eq:hcubicBorelblock} follows directly from the transport matrices in Theorem \ref{thm:hcubicblock}. The Abel transform shifts the cubic density exponents
\[
-\frac23,\,-\frac13 \quad \rightarrow\quad 
-\frac16,\,\frac16
\]
and multiplies the corresponding monodromy eigenvalues by $-1$. Thus the density monodromy $M$ has order three, while $-M$ has order six.
\end{proof}

%\subsubsection{Uniform cubic confluence}

Fix $a>0$ and $j\ge1$. In the cubic normal form of Lemma \ref{lem:hcubicnormal},
\[
F(q;C)=F_c(\mu)+\Lambda(\mu)u-\frac{u^3}{3},
\qquad q=q(u,\mu),
\]
set
\[
H(u,\mu):=i\frac{q(u,\mu)}{\sin q(u,\mu)}\,q_u(u,\mu),
\qquad
H_0:=H(0,0)=\frac{i\theta_j}{\gamma_j\sin\theta_j}\ne0.
\]
Since $\Lambda'(0)>0$, we use $\Lambda$ as the local parameter and write $\mu=\mu(\Lambda)$. Set
\[
\Lambda=\varepsilon^2,\qquad u=\varepsilon v,
\qquad P(v)=v-\frac{v^3}{3}.
\]
Then
\[
F-F_c=\varepsilon^3P(v), \qquad F_u=\varepsilon^2-u^2.
\]

For $X$ near $F_c(\mu)$, let $u_\ell(X,\mu)$, $\ell=0,1,2$, denote the three local inverse branches. For a fixed coefficient vector $\mathbf n=(n_0,n_1,n_2)$, define the corresponding pushed density by
\[
\mathcal Q_{\mathbf n,\mu}(X)
:=\sum_{\ell=0}^2n_\ell\,\frac{H(u_\ell(X,\mu),\mu)}{\Lambda(\mu)-u_\ell(X,\mu)^2}.
\]
If $v_\ell(s)$ are the three inverse branches of $P(v)=s$, set
\[
R_{\mathbf n}(s):=\sum_{\ell=0}^2\frac{n_\ell}{1-v_\ell(s)^2}.
\]

This scaling and the Abel transform give the following uniform confluence formulas.

\begin{proposition}
\label{prop:hcubicconfluence}
For fixed $\mathbf n$ and sufficiently small $\varepsilon$:
\begin{enumerate}[label=\textnormal{(\roman*)}]
\item The critical points and values are
\[
u_\pm=\pm\varepsilon,\qquad
D_\pm=F_c\pm\frac23\varepsilon^3,
\]
and, locally uniformly away from $s=\pm2/3$,
\[
\varepsilon^2\mathcal Q_{\mathbf n,\mu}(F_c+\varepsilon^3s)=H_0R_{\mathbf n}(s)+O(\varepsilon).
\]

\item  Let $\omega_{\mu}$ be the action difference and $\widehat S_\mu$ be the corresponding ordinary Borel germ. Along a fixed lifted path avoiding $s=\pm2/3$,
\[
\varepsilon^{1/2}\widehat S_\mu(\omega_\mu+\varepsilon^3s)
=\frac{H_0}{\sqrt\pi}\int_{-\infty}^s
 \frac{R_{\mathbf n}(z)}{\sqrt{s-z}}\dd z+O(\varepsilon^{1/2}).
\]
The action separation, density scale and ordinary Borel scale are $\Lambda^{3/2}$, $\Lambda^{-1}$, $\Lambda^{-1/4}$, respectively.

\item At resonance, every nonzero integral reduced cycle has
\[
\mathcal Q_{\mathbf n,0}(F_c-\delta)
=c_{\mathbf n}\delta^{-2/3}+O(\delta^{-1/3}),
\qquad c_{\mathbf n}\ne0,
\]
and consequently
\[
\widehat S_0(\xi)=C_{\mathbf n}(\omega_0-\xi)^{-1/6}+O(1),\qquad C_{\mathbf n}\ne0.
\]
\end{enumerate}
\end{proposition}

\noindent The nonvanishing in \textnormal{(iii)} follows because no nonzero integral reduced cycle cancels the leading cubic term.

\subsection{Complete directional Stokes operators}

Assume throughout this subsection that $C$ is nonresonant. We combine the local connection formulas along arbitrary action rays.

\vskip 0.1cm
We first establish convergence of the lateral Laplace integrals.

\begin{lemma}
\label{lem:raygrowth}
Fix an affine ray $X=A+se^{i\theta}$ and a finite incoming continuation. Every continued inverse branch on the tail of the ray satisfies
\[
|q(X)|\le K(1+|X|),\qquad |F'(q(X))|\ge c>0 .
\]
Consequently, for each fixed $n\ge1$,
\[
\left|\frac{W(q(X))^n}{F'(q(X))}\right|=O((1+s)^n).
\]
Hence the complete lateral integrals $\Psi_A^{\theta,\pm}(t)$ converge whenever $\re\!\left(\frac{e^{i\theta}}{t}\right)>0$.
\end{lemma}

\begin{proof}
Remove a bounded initial segment of the ray so that it contains no critical values. Along an inverse branch $F(q)=X$, with $X=A+se^{i\theta}$, we obtain
\[
aqF'(q)=-P_X(q),\qquad
P_X(q)=(a^2+y^2)q^2-2yXq+X^2-aX,
\]
and
\[
W(q)^2=q^2+\frac{(X-yq)^2}{a^2}.
\]

If $|q|\ge M(1+|X|)$ for $M$ sufficiently large, then
$|P_X(q)|\ge c|q|^2$, hence $|F'(q)|\ge c|q|$. Since
$q'(s)=e^{i\theta}/F'(q(s))$, this implies
$\frac{\dd}{\dd s}|q(s)|^2=O(1)$ on such intervals. As $|X(s)|\asymp s$, we obtain
\[
|q(X)|\le K(1+|X|).
\]

Suppose that $F'(q(X))$ is not bounded away from zero. Then for some $X_j\to\infty$ on the ray,
$F'(q_j)\to0$, where $q_j=q(X_j)$. Since
$P_{X_j}(q_j)=o(|q_j|)$ and $q_j=O(|X_j|)$, after passing to a subsequence,
\[
\frac{q_j}{X_j}\longrightarrow\frac1{y-ia}
\quad\text{or}\quad\frac1{y+ia}.
\]
Then the identity
\[
e^{2iq}=\frac{X-(y-ia)q}{X-(y+ia)q}
\]
gives $\im q_j\to+\infty$ or $\im q_j\to-\infty$. In either case $F''(q_j)$ tends to a nonzero limit. Hence, by the local inverse theorem applied to $F'$, there are critical points $z_j$ with $z_j-q_j\to0$ and therefore $F(z_j)-X_j\to0$.

This contradicts the transverse asymptotics of the critical values in Proposition~\ref{prop:spectrum}: their distance from any fixed affine line tends to infinity along each of the four tails. Thus $|F'(q(X))|\ge c>0$ on the tail.

Finally, the preceding identities give $W(q(X))=O(1+s)$, and hence
\[
\left|\frac{W(q(X))^n}{F'(q(X))}\right|
=O((1+s)^n).
\]
Laplace convergence follows from the exponential factor $e^{-se^{i\theta}/t}$ whenever $\re(e^{i\theta}/t)>0$.
\end{proof}

\iffalse
\begin{proof}
The asymptotic root estimates of Lemma \ref{lem:hcomplexnoescape} give $|q(X)|=O(1+|X|)$ along the ray. On an inverse root,
\[
F'(q)=-\frac{P_X(q)}{aq}, \qquad
 P_X(q)=(a^2+y^2)q^2-2yXq+X^2-aX,
\]
and
\[
W(q)^2=q^2+\frac{(X-yq)^2}{a^2}.
\]
The two roots of $P_X$ are asymptotic to $X/(y\mp ia)$. If $F'(q(X))$ approached $0$, the inverse branch would approach one of these roots, contradicting the corresponding large-$X$ inverse asymptotics. Thus $|F'(q(X))|$ is bounded below, and the last identity gives the polynomial density bound. Exponential decay along the ray then gives the stated Laplace convergence.
\end{proof}
\fi

We write the directional Stokes relation as
\[
\boldsymbol\Psi^{\theta,+}
=S_\theta\,\boldsymbol\Psi^{\theta,-}.
\]

\begin{theorem}
\label{thm:directional}
For every direction $\theta$ and every $A\in\Sigma_F$,
\begin{equation}\label{eq:directionalStokes}
\Psi_A^{\theta,+}=\Psi_A^{\theta,-}
+\sum_{B\in\Sigma_F\cap(A+e^{i\theta}\R_+)}
s_{AB}^{\theta}\Psi_B^{\theta,-},
\qquad
s_{AB}^{\theta}\in\{-2,-1,0,1,2\}.
\end{equation}
The sum is finite. Moreover, $S_\theta$ is row-finite and unipotent on each affine-line block; hence $S_{\theta}^{-1}$ and $\Lambda_\theta:=\log S_\theta$ are also row-finite.
\end{theorem}

\begin{proof}
Change the detours from lower to upper in increasing action order. Since the source cycle is a difference of two inverse sheets, the Picard--Lefschetz coefficient at each crossing is $n_j-n_i\in\{0,\pm1,\pm2\}$. Proposition~\ref{prop:hMorseconnection} gives the local jumps; their sum is \eqref{eq:directionalStokes}. The sum is finite by Proposition~\ref{prop:spectrum}, and the lateral Laplace integrals converge by Lemma \ref{lem:raygrowth}.

On each affine-line block the strict part $S_\theta-I$ is finite and upper triangular in action order, hence nilpotent. The remaining claims follow.
\end{proof}

For $\omega\in e^{i\theta}\R_+$, define the alien operator by
\begin{equation}\label{eq:alien}
\Delta_\omega\Phi_A:=(\Lambda_\theta)_{A,A+\omega}\Phi_{A+\omega},
\end{equation}
with the right-hand side understood as zero if $A+\omega\notin\Sigma_F$.

The multiplicative Stokes action extends to finite polynomial expressions in the labelled saddle sectors, and the alien operators $\Delta_\omega$ satisfy the action-graded Leibniz rule \cite[Sections~29--30]{Sauzin}.

\subsection{The positive-ray Stokes matrix and logarithm}
Suppose $\theta_m<C<\theta_{m+1}$, with $m=0$ interpreted as $0<C<\theta_1$. Order the real-action basis as
\[
(A_1,B_1,A_2,\ldots,B_m,A_{m+1}).
\]
We orient the vanishing cycles by
\[
\delta_{A_k}=[q_{\rm low}]-[q_{\rm up}],\qquad
\delta_{B_k}=[q_{\rm L}]-[q_{\rm R}],
\]
at maxima and minima, respectively.

\begin{theorem}\label{thm:positive}
In these orientations, the full positive-direction Stokes matrix is
\begin{equation}\label{eq:positiveS}
 (S_0)_{ij}=
 \begin{cases}
 0,&\quad i>j,\\
 1,&\quad i=j,\\
 2,&\quad i=2k-1,\ j=2k,\ 1\le k\le m,\\
 1,&\quad i<j~\text{ in every other case}.
 \end{cases}
\end{equation}
The formula is constant throughout each nonresonant real parameter chamber. Its logarithm has no strict $A$-to-$A$ or $B$-to-$B$ entries.
\end{theorem}

Figure \ref{fig:positive-ray-stokes} shows the Picard--Lefschetz coefficients: $2$ at the first collision of a sector based at $A_k$, and $1$ at each later collision.

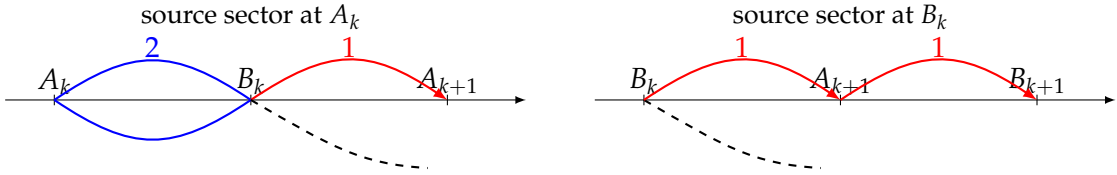
\begin{figure}[h]
\centering
\begin{tikzpicture}[x=1.3cm,y=1.0cm,>=latex]

% ---------- left panel ----------
\node at (2,2.1) {\small source sector at $A_k$};
\draw[->] (-0.5,1) -- (4.8,1);
\node at (0,1.25) {$A_k$};
\node at (2,1.25) {$B_k$};
\node at (4,1.25) {$A_{k+1}$};
\foreach \x in {0,2,4} \draw (\x,1.07)--(\x,0.93);

\draw[thick,blue] (0,1) .. controls (0.8,1.7) and (1.2,1.7) .. (2,1);
\draw[thick,blue] (0,1) .. controls (0.8,0.3) and (1.2,0.3) .. (2,1);
\node[blue] at (1,1.7) {$2$};

\draw[thick,red,->] (2,1) .. controls (2.8,1.7) and (3.2,1.7) .. (4,1);
\draw[thick,dashed] (2,1) .. controls (2.8,0.4) and (3.1,0.15) .. (3.8,0.1);
\node[red] at (3,1.7) {$1$};

% ---------- right panel ----------
\node at (8,2.1) {\small source sector at $B_k$};
\draw[->] (5.5,1) -- (10.8,1);
\node at (6,1.25) {$B_k$};
\node at (8,1.25) {$A_{k+1}$};
\node at (10,1.25) {$B_{k+1}$};
\foreach \x in {6,8,10} \draw (\x,1.07)--(\x,0.93);

\draw[thick,dashed] (6,1) .. controls (6.8,0.4) and (7.1,0.15) .. (7.8,0.1);
\draw[thick,red,->] (6,1) .. controls (6.8,1.7) and (7.2,1.7) .. (8,1);
\draw[thick,red,->] (8,1) .. controls (8.8,1.7) and (9.2,1.7) .. (10,1);

\node[red] at (7,1.7) {$1$};
\node[red] at (9,1.7) {$1$};

\end{tikzpicture}
\caption{Picard--Lefschetz coefficients along the positive ray.}
\label{fig:positive-ray-stokes}
\end{figure}

\begin{proof}
By Proposition \ref{prop:hlabelinduction}, the sector based at $A_k$ reaches $B_k$ with the two colliding sheets carrying coefficients $+1$ and $-1$. Hence the Picard--Lefschetz coefficient at this first collision is
\[
n_j-n_i=2,
\]
with the chosen orientation. Thus
\[
s_{A_k,B_k}=2.
\]

After crossing $B_k$, one of these sheets follows the pole-side branch and does not meet any later positive critical value, while the other continues along the increasing branch. At each later collision this branch has coefficient $\pm1$ and the newly encountered sheet has coefficient $0$. Therefore
\[
|n_j-n_i|=1,
\]
and our orientation convention gives coefficient $+1$. A sector based at $B_k$ has one pole-side branch and one increasing branch, so the same calculation applies at every later collision.

Thus, in the action ordering
\[
 A_1<B_1<A_2<\cdots<B_m<A_{m+1},
\]
the only exceptional off-diagonal entries are $s_{A_k,B_k}=2$; all other later transitions have coefficient $1$. Earlier actions are not crossed, so the entries below the diagonal are $0$; the diagonal entries are $1$. This proves \eqref{eq:positiveS}.
\end{proof}

To compute the logarithm, pair the labels as $(A_k,B_k)$ and let the formal variable $z$ encode a shift by one pair. The corresponding block symbol of $S_0$ is
\begin{equation}\label{eq:symbol}
S(z)=\frac{1}{1-z}
 \begin{pmatrix}
  1&2-z\\
  z&1
 \end{pmatrix}.
\end{equation}
Set
\[
B(z)=
 \begin{pmatrix}
  0&2-z\\
  z&0
 \end{pmatrix},
 \qquad u^2=z(2-z).
\]
Then
\[
S(z)=\frac{I+B(z)}{1-z}, \qquad B(z)^2=u^2I.
\]
Hence, formally,
\[
\log(I+B)=\frac12\log(1-u^2)\,I+
\frac{\operatorname{artanh}u}{u}\,B=\log(1-z)\,I+
\frac{\operatorname{artanh}u}{u}\,B.
\]
Therefore
\begin{equation}\label{eq:logsymbol}
\log S(z)=h(z)B(z),\qquad
h(z):=\frac{\operatorname{artanh}\sqrt{z(2-z)}}
      {\sqrt{z(2-z)}}.
\end{equation}

Extracting coefficients gives the following formula.
\begin{proposition}
\label{prop:positiveLog}
Let $\Lambda_0=\log S_0$ and $h(z)$ be as above. Then, whenever the indicated labels occur,
\begin{align*}
(\Lambda_0)_{A_k,B_l}&=[z^{\,l-k}](2-z)h(z),\qquad l\ge k,\\
(\Lambda_0)_{B_k,A_l}&=[z^{\,l-k-1}]h(z),\qquad l>k,\\
(\Lambda_0)_{A_k,A_l}&=(\Lambda_0)_{B_k,B_l}=0,\qquad l>k.
\end{align*}
Here $[z^n]$ denotes coefficient extraction.
\end{proposition}

\begin{example}
\label{ex:primitivecancellation}
If $\theta_1<C<\theta_2$, then
\[
 S_0=
 \begin{pmatrix}
  1&2&1\\
  0&1&1\\
  0&0&1
 \end{pmatrix},
 \qquad
 \Lambda_0=
 \begin{pmatrix}
  0&2&0\\
  0&0&1\\
  0&0&0
 \end{pmatrix}.
\]
Hence
\[
 \Delta_{B_1-A_1}\Phi_{A_1}=2\Phi_{B_1},
 \qquad
 \Delta_{A_2-B_1}\Phi_{B_1}=\Phi_{A_2},
 \qquad
 \Delta_{A_2-A_1}\Phi_{A_1}=0.
\]
Nevertheless, $A_2-A_1$ is a singularity of the continued Borel germ, since
\[
(S_0)_{13}=\frac12(\Lambda_0^2)_{13}=1.
\]
Thus the full Borel singular support can be strictly larger than the primitive alien support.
\end{example}

\subsection{Sectorial completion of countably many actions}
We now combine the directional Stokes matrices in a fixed angular sector. There may be infinitely many Stokes directions, but each matrix entry involves only finitely many intermediate actions under the order defined below. For a sheaf-theoretic framework for resurgent continuation and Stokes data, see \cite{KapranovSoibelman}.

\begin{definition}
Let $\Gamma\subset\C$ be a closed angular cone of aperture strictly less than $\pi$, with vertex $0$ and bisector angle $\phi$. Set $h_\phi(z):=\re(e^{-i\phi}z)$. Then
\[
 h_\phi(z)\ge c_\Gamma |z|,
 \qquad \text{ for }z\in\Gamma~\text{and some }c_\Gamma>0.\]

For $A,B\in\Sigma_F$, write
\[
A\preceq_\Gamma B
\quad\Longleftrightarrow\quad
B-A\in\Gamma.
\]
The corresponding action interval is
\[
[A,B]_\Gamma:=\Sigma_F\cap(A+\Gamma)\cap(B-\Gamma).
\]
\end{definition}
The interval lies in the disk $|X-A|\le h_{\phi}(B-A)/c_{\Gamma}$, so Proposition \ref{prop:spectrum} gives
\begin{equation}\label{eq:intervalcount}
\#[A,B]_\Gamma\leq K_0+\frac{K_1}{c_{\Gamma}}h_{\phi}(B-A)   
\end{equation}

\begin{definition}
Let $I_\Gamma(\Sigma_F)$ be the space of matrices $K=(K_{AB})_{A,B\in\Sigma_F}$ supported on $A\preceq_\Gamma B$, with
\[
(KL)_{AB}=\sum_{C\in[A,B]_\Gamma}K_{AC}L_{CB}.
\]
Give it the topology of coefficientwise convergence. The diagonal-one matrices form a group under this product, and their logarithms are defined coefficientwise by finite chain sums.
\end{definition}

For each ray
\[
 \ell=e^{i\theta}\R_+\subset\Gamma,
\]
regard the directional Stokes matrix $S_\theta$ of Theorem \ref{thm:directional} as an element of $I_\Gamma(\Sigma_F)$, and write it as $S_\ell$. Multiply these matrices in increasing-angle order. Only countably many factors can differ from the identity, since their nonzero entries correspond to differences $B-A$.

\begin{theorem}\label{thm:angular}
The angle-ordered product
\begin{equation}\label{eq:angularproduct}
P_\Gamma=\prod_{\ell\subset\Gamma}^{\curvearrowright}S_\ell
\end{equation}
is well defined in $I_\Gamma(\Sigma_F)$. There are constants $C_\Gamma',\kappa_\Gamma>0$ such that
\begin{equation}\label{eq:coeffbound}
 |(P_\Gamma)_{AB}|\le C_\Gamma'e^{\kappa_\Gamma h_\phi(B-A)}\qquad(A\prec_\Gamma B).
\end{equation}
For sufficiently large $\sigma$, it defines a bounded invertible operator on
\begin{equation}\label{eq:banach}
 \mathcal H_{\sigma,\phi}
 =\left\{(c_A)_{A\in\Sigma_F}:\
       \sum_A|c_A|e^{-\sigma h_\phi(A)}<\infty\right\}.
\end{equation}
\end{theorem}

\begin{proof}
Fix $A\preceq_\Gamma B$. Every chain contributing to the $(A,B)$-entry lies in the finite interval $[A,B]_\Gamma$. If $M=\#[A,B]_\Gamma$, then each edge coefficient has modulus at most $2$, and therefore
\[
|(P_\Gamma)_{AB}|\le\sum_{k=1}^{M-1}
 2^k\binom{M-2}{k-1}=2\cdot3^{M-2}.
\]
Together with \eqref{eq:intervalcount}, this gives \eqref{eq:coeffbound} and proves that the angle-ordered product is coefficientwise well defined.

On $\mathcal H_{\sigma,\phi}$, we have
\[
\|K\|\le \sup_A\sum_B |K_{AB}|e^{-\sigma h_\phi(B-A)}.
\]
Uniform separation, linear counting, and \eqref{eq:coeffbound} imply that the strict part $P_\Gamma-I$ tends to zero in this norm as $\sigma\to\infty$. Hence, for sufficiently large $\sigma$, $P_\Gamma$ is bounded and invertible, and the logarithmic series for $\log P_\Gamma$ converges in operator norm.
\end{proof}

\begin{remark}
Subdividing the sector factors $P_\Gamma$ into ordered products. Since the factors need not commute, $\log P_\Gamma$ includes Baker--Campbell--Hausdorff terms in addition to the directional logarithms.   
\end{remark}

At $C=\theta_j$, the local ordinary Borel monodromy has order six and is not unipotent, so the logarithmic construction above does not apply directly. Moreover, the separation constants in the weighted completion degenerate as the two critical values coalesce.

\section{Borel connection and contour reconstruction}\label{sec:hconnection}

We now determine how the original real Fourier contour is expressed in the saddle sectors of Section \ref{sec:hglobal}. First we decompose the contour into compact action intervals and a terminal ray. We then compute the scalar connection across the first pole--saddle block.

Throughout, the Fourier phase is written as $\tau\mapsto F(-i\tau)$.

\subsection{Action-space reconstruction of the Fourier contour}
For $k\ge1$, set
\[
\mathfrak C_k(t;r,y)
=2\pi i\Res_{\tau=ik\pi}\bigl(e^{-F(-i\tau)/t}V(\tau)\bigr),
\qquad H_N=\left(N+\tfrac12\right)\pi.
\]
If $N\pi<H<(N+1)\pi$, contour translation gives
\begin{equation}\label{eq:hfinite-residue}
\int_\R e^{-F(-i\tau)/t}V\dd\tau
=\int_{\R+iH}e^{-F(-i\tau)/t}V\dd\tau
+\sum_{k=1}^N\mathfrak C_k(t;r,y).
\end{equation}
At the half-integer heights $H_N$,
\[
\re F(H_N-ix)=H_Ny+a x\tanh x\ge H_Ny.
\]

\begin{proposition}\label{thm:hglobalres}
For $r\ge0$ and $y,t>0$,
\begin{equation}\label{eq:hglobalres}
p_t(r,y)=\frac1{(2\pi t)^2}\sum_{k\ge1}\mathfrak C_k(t;r,y),
\end{equation}
and the series is absolutely convergent.
\end{proposition}
\begin{proof}
On $\R+iH_N$, the preceding identities imply
\[
\left|\int_{\R+iH_N}e^{-F(-i\tau)/t}V\dd\tau\right|
\le e^{-H_Ny/t}\int_\R\frac{\sqrt{x^2+H_N^2}}{\cosh x}\dd x\le(\pi H_N+4)e^{-H_Ny/t}.
\]
Consequently,
\[
\left|p_t(r,y)-\frac1{(2\pi t)^2}
\sum_{k=1}^N\mathfrak C_k(t;r,y)\right|
\le \frac{\pi H_N+4}{(2\pi t)^2}e^{-H_Ny/t}.
\]
The vertical sides of the translating rectangles vanish because
\[V(x+iH)=O_H(|x|e^{-|x|}) \qquad\re(\tau\coth\tau)=|\re\tau|+o(1).\] 
The residue theorem gives \eqref{eq:hfinite-residue}; letting $N\to\infty$ proves the claim. Applying the same argument between two consecutive half-integer heights gives absolute convergence.
\end{proof}

At $r=0$, one recovers the classical vertical-axis formula
\begin{equation}\label{eq:hsech}
p_t(0,y)=\frac1{2t^2}\sum_{k\ge1}(-1)^{k+1}k e^{-k\pi|y|/t}
=\frac1{8t^2}\sech^2\!\left(\frac{\pi|y|}{2t}\right).
\end{equation}
This formula and the vertical-axis small-time asymptotics are classical; see \cite{CL,Gaveau,Hulanicki,Li2007,LZ}.

For the action-space decomposition, assume first that $C$ is nonresonant and set
\[
m=\#\{k:\theta_k<C\},\qquad A_*=A_{m+1}.
\]
Whenever $F(q)=X$ has a conjugate nonreal pair $q_{\mathrm{up}}(X),q_{\mathrm{down}}(X)$, write
\[
\rho(X)=iW(q_{\mathrm{down}}(X))q_{\mathrm{down}}'(X)-iW(q_{\mathrm{up}}(X))q_{\mathrm{up}}'(X).
\]
Let $\rho_k$ denote this density on $(A_k,B_k)$, $1\le k\le m$, and let $\rho_*$ denote the corresponding density on the terminal ray $(A_*,\infty)$.

\begin{proposition}
\label{thm:hfiniteactionreconstruction}
For $r,y,t>0$ and nonresonant $C$,
\begin{equation}\label{eq:v3finitephysical}
(2\pi t)^2p_t(r,y)=
\sum_{k=1}^{m}\int_{A_k}^{B_k}e^{-X/t}\rho_k(X)\dd X
+\int_{A_*}^{\infty}e^{-X/t}\rho_*(X)\dd X.
\end{equation}
Moreover,
\[
\int_{A_k}^{B_k}e^{-X/t}\rho_k(X)\dd X=
\mathfrak C_k(t;r,y),
\qquad 1\le k\le m.
\]
All integrals are absolutely convergent. If $C=\theta_j$, the same formula holds with the last compact interval replaced by $[A_j,\mathscr D_j]$ and the terminal ray starting at $\mathscr D_j$.
\end{proposition}

\begin{proof}
For $A_k<X<B_k$, the two inverse branches form a closed $q$-contour enclosing only the pole $k\pi$. Under $\tau=iq$ this gives
\[
\int_{A_k}^{B_k} e^{-X/t}\rho_k(X)\dd X=
\mathfrak C_k(t;r,y).
\]
The terminal pair defines an unbounded contour in the upper $\tau$-plane. Closing it to the real axis encloses exactly the first $m$ poles; the closing pieces vanish by the outer asymptotics. The residue theorem and the change of variable $X=F(-i\tau)$ then give \eqref{eq:v3finitephysical}. At resonance the same argument applies, since
\[
\rho(X)=O(|X-\mathscr D_j|^{-2/3})
\]
is locally integrable.
\end{proof}

\begin{remark}
If $0<C<\theta_1$, there are no compact action intervals, so the positive Borel sum of the minimal action germ is the full heat kernel. At $C=\theta_1$, the first interval ends at the cubic action and both lateral continuations must be retained.
\end{remark}

\subsection{Local pole--saddle coordinate and scalar connection}\label{sec:hscalarconnection}
Fix $y>0$ and take $r>0$ small enough that the first pole $i\pi$ lies between two simple critical points $\tau_-$ and $\tau_+$. Write
\[
A_1=F(-i\tau_-),\qquad B_1=F(-i\tau_+),\qquad
A_c=\frac{B_1+A_1}{2},\qquad d=\frac{B_1-A_1}{2}>0.
\]
Chang--Li \cite{CL} use this saddle--pole configuration in their uniform Bessel analysis. We use the following exact Joukowski coordinate.

\begin{lemma}\label{prop:hnormal}
There is a holomorphic coordinate $w$ on a neighborhood of the closed unit disk such that
\begin{equation}\label{eq:hnormal}
F(-i\tau(w))=A_c-\frac d2(w+w^{-1}),\qquad
\tau(0)=i\pi,\quad \tau(1)=\tau_-,\quad \tau(-1)=\tau_+.
\end{equation}
Moreover
\[
H(w):=wV(\tau(w))\tau'(w)
\]
extends holomorphically across $w=0$ and satisfies $H(0)=-i\pi$. The image of $|w|=1$ winds once counterclockwise around $i\pi$.
\end{lemma}

\begin{proof}
Set $b=\sqrt{\pi a/y}$ and
\[
P_b(u)=\frac{2(A_c-F(\pi-bu))}{d},\qquad J(u)=u+u^{-1}.
\]
The Laurent expansion at the pole gives $uP_b(u)=u^2+1+O(b^2)$ uniformly on a fixed disk, while the two critical points are $u_\pm=\pm1+O(b^2)$. Hence
\[
(uP_b(u))^2-4u^2=(u-u_+)^2(u-u_-)^2A_b(u),\qquad A_b(u)=1+O(b^2),
\]
with $A_b$ holomorphic and nonzero. Choosing the square root of $A_b$ near one gives a holomorphic branch $W_b(u)=u+O(b^2u)$ satisfying $J(W_b(u))=P_b(u)$. Its local inverse defines
$\tau(w)=i\pi-ibW_b^{-1}(w)$ and yields \eqref{eq:hnormal}. Finally $V(i\pi+s)=-i\pi/s+O(1)$ and $s=w\gamma(w)$ with $\gamma(0)\ne0$, so $H$ extends holomorphically with $H(0)=-i\pi$;
the winding statement follows from the argument principle.
\end{proof}

Set
\[
 \zeta:=\frac{X-A_1}{d},
 \qquad
 \mathcal N:=\pi\sqrt{2\pi}.
\]
Thus the first compact action interval $[A_1,B_1]$ becomes $0\le\zeta\le2$.

Let $\rho_1$ be the density on $(A_1,B_1)$ from Proposition~\ref{thm:hfiniteactionreconstruction}, and let $\rho_{\rm out}$ be the outgoing two-sheet density after $B_1$, with the orientation fixed in Section~\ref{sec:hglobal}. Define
\[
b_-(\zeta):=\frac{d}{\mathcal N}\rho_1(A_1+d\zeta),
 \qquad
b_+(\eta):=-\frac{id}{\mathcal N}\rho_{\rm out}(B_1+d\eta).
\]
Write $b_-^{0\pm}$ for the upper and lower continuations of $b_-$ through $\zeta=2$.

The local pole--saddle coordinate gives (see the proof of Theorem \ref{thm:hcomplete})
\begin{equation}\label{eq:hlocaldisc}
\Disc_2 b_-(2+\eta)
:=b_-^{0+}(2+\eta)-b_-^{0-}(2+\eta)
=2i\,b_+(\eta).
\end{equation}
At $B_1$, the incoming density has jump $2i$ times the outgoing saddle density. The corresponding action difference is
\[
B_1-A_1=2d.
\]
Continue the outgoing branches along the matched upper and lower paths of Section \ref{sec:hglobal}, and set
\begin{align*}
U_-^\pm(t)&=e^{-A_1/t}
 \int_0^\infty
 e^{-d\zeta/t}b_-^{0\pm}(\zeta)\dd\zeta,\\
U_+^\pm(t)&=e^{-B_1/t}
 \int_0^\infty
 e^{-d\eta/t}b_+^\pm(\eta)\dd\eta.
\end{align*}
By Lemma \ref{lem:raygrowth}, these integrals converge for $\re(1/t)>0$.

\begin{theorem}
\label{thm:hcomplete}
For nonresonant $C>\theta_1$ and $\re(1/t)>0$,
\begin{align}
 U_-^+-U_-^-
 &=i(U_+^++U_+^-),
 \label{eq:hcomplete}\\
 U_-^+-iU_+^+
 &=U_-^-+iU_+^-=
 \frac{\mathfrak C_1(t;r,y)}{\mathcal N}.
 \notag
\end{align}
Consequently, for $\re t>0$,
\begin{equation}\label{eq:hphysicalcomplete}
p_t(r,y)=\frac{\mathcal N}{(2\pi t)^2}
(U_-^+-iU_+^+)+\frac{1}{(2\pi t)^2}
\sum_{k\ge2}\mathfrak C_k(t;r,y).
\end{equation}
\end{theorem}

\begin{proof}
For $C>\theta_1$, the first pair $A_1<B_1$ is Morse. The Joukowski coordinate of Lemma~\ref{prop:hnormal} fixes the local orientation and, with the normalization above, gives
\[
b_-^{0+}(2+\eta)=i\,b_+(\eta),
 \qquad
b_-^{0-}(2+\eta)=-i\,b_+(\eta).
\]
Equivalently,
\[
\Disc_2 b_-(2+\eta)=2i\,b_+(\eta).
\]
Lemma~\ref{prop:hnormal} computes this coefficient near the pole. The coefficient belongs to the same ordered two-sheet Picard--Lefschetz cycle, whose branches and orientation are preserved by Proposition~\ref{prop:hlabelinduction}. Thus these identities hold for every nonresonant $C>\theta_1$.

On $0<\zeta<2$ the upper and lower continuations coincide. By Proposition~\ref{thm:hfiniteactionreconstruction},
\[
e^{-A_1/t}\int_0^2 e^{-d\zeta/t}b_-(\zeta)\,d\zeta
=\frac{\mathfrak C_1(t;r,y)}{\mathcal N}.
\]
Splitting the complete integrals at $\zeta=2$ and using
\[
e^{-A_1/t}e^{-2d/t}=e^{-B_1/t},
\]
together with the two continuation identities above, gives
\[
U_-^+-U_-^-=i\bigl(U_+^++U_+^-\bigr),
\]
and
\[
U_-^+-iU_+^+=U_-^-+iU_+^-
=\frac{\mathfrak C_1(t;r,y)}{\mathcal N}.
\]
Finally, Proposition~\ref{thm:hglobalres} gives
\[
p_t(r,y)=\frac{\mathcal N}{(2\pi t)^2}\bigl(U_-^+-iU_+^+\bigr)
+\frac{1}{(2\pi t)^2}\sum_{k\ge2}\mathfrak C_k(t;r,y),
\]
which proves the reconstruction formula.
\end{proof}

The scalar connection extends through a cubic resonance $C=\theta_j$, $j\ge2$: the density is $O(|X-\mathscr D_j|^{-2/3})$, so a detour of radius $\varepsilon$ contributes $O(\varepsilon^{1/3})$. The reconstruction formulas therefore hold along the prescribed lateral cubic detours.

\section{Boundary regimes: horizontal endpoints and isotropic complex endpoints}
\label{sec:hboundary}
We consider two boundary regimes:
\begin{itemize}
\item horizontal endpoints, where finite complex actions exist but the positive Borel ray is singularity-free;
\item nonzero isotropic complex endpoints, which admit no finite normal trajectories but retain the analytically continued kernel and the limiting vertical action spectrum.
\end{itemize}

\subsection{Horizontal endpoints: complex actions and Borel summation}
\label{sec:hhorizontal}

Let $y=0$ and $a=r^2/2>0$. The phase is even,
\[
F(-i\tau)=a\tau\coth\tau=a+\frac{a}{3}\tau^2+O(\tau^4),
\]
and therefore
\begin{equation}\label{eq:hhorizontalexp}
p_t(r,0)\sim e^{-a/t}t^{-3/2}\sum_{n\ge0}A_n(r)t^n .
\end{equation}
On the positive real axis,
\[
u=F(-i\tau)-a
\]
is strictly increasing for $\tau>0$. Writing its inverse as $\tau_+(u)$ gives the exact Laplace representation
\begin{equation}\label{eq:hhorizontalLaplace}
p_t(r,0)=e^{-a/t}t^{-2}\int_0^\infty e^{-u/t}\rho_r^{\rm hor}(u)\dd u,\qquad
\rho_r^{\rm hor}(u)=\frac{1}{2\pi^2}\frac{\tau_+(u)}{\sinh\tau_+(u)}
\frac{\dd\tau_+}{\dd u}.
\end{equation}

To remove the symmetry $\tau\leftrightarrow-\tau$, introduce
\[
s=\tau^2, \qquad \mathfrak F(s)=a\bigl(\sqrt{s}\coth\sqrt{s}-1\bigr),
\qquad \mathfrak F'(0)=\frac a3.
\]
Thus $s=0$ is a regular inverse point. Define
\[
\Omega_a:=\left\{a\sinh^2\tau_c:
\tau_c\ne0,~\sinh\tau_c\cosh\tau_c=\tau_c\right\},
\]
where the pair $\{\tau_c,-\tau_c\}$ is counted once.

\begin{proposition}
\label{prop:hhorizontalspectrum}
Every nonzero critical point of $\mathfrak F$ is Morse. The set $\Omega_a$ is infinite, locally finite, pairwise distinct, and 
\[
\Omega_a\subset\C\setminus\R,\qquad
\Omega_a=\bigl(\Sigma_H^{\rm fin}(r,0)-a\bigr)\setminus\{0\}.
\]
Moreover, every inverse germ of $\mathfrak F$ continues along every finite path in $\C\setminus\Omega_a$, with no escape to infinity or to a phase pole over a compact action path.
\end{proposition}

\begin{proof}
At a nonzero critical point, $\sinh\tau_c\cosh\tau_c=\tau_c$, and direct differentiation gives
\[
\frac{\dd^2}{\dd\tau^2}F(-i\tau)\bigg|_{\tau=\tau_c}=2a,
\]
so all such points are Morse. Writing $\Delta=a\sinh^2\tau_c$ gives
\[
s_c=\frac{\Delta}{a}\left(1+\frac{\Delta}{a}\right).
\]
Equal critical values therefore give the same quotient critical point, and bounded critical values give bounded $s_c$. This proves distinctness and local finiteness. A real nonzero critical value would force $\tau_c$ to be real or purely imaginary, contradicting
\[
\sinh(2x)=2x,\qquad\sin(2x)=2x
\]
for nonzero real $x$. The contraction argument used in Lemma~\ref{lem:hremotecritical}, applied near
\[
\frac12\log(4\pi n)+i\left(\pi n+\frac{\pi}{4}\right)
\]
gives infinitely many critical points. The action identity follows from Theorem \ref{thm:hactionclass} with $y=0$.

For inverse continuation, the equation in the variable $q=-i\tau$ is
\[
 a\cos q-(a+u)\frac{\sin q}{q}=0.
\]
The same Rouch\'e estimate as in Lemma \ref{lem:hcomplexnoescape} gives one simple root near $(n+\tfrac12)\pi$ for $|n|$ sufficiently large, and no other roots with sufficiently large modulus; more precisely,
\[
q_n(u)=\left(n+\frac12\right)\pi-\frac{a+u}{a(n+\frac12)\pi}+O(n^{-2}).
\]
Thus inverse lifts cannot escape to infinity or reach a pole at bounded action. The origin is regular in the quotient coordinate, since $\mathfrak F'(0)=a/3\ne0$.
\end{proof}

\begin{theorem}
\label{thm:hhorizontalresurgence}
Let $S_r^{\rm hor}(t):=\sum_{n\ge0}A_n(r)t^n$ be the series in \eqref{eq:hhorizontalexp}. Then
\[
\bigcup_{\mathfrak p}\operatorname{Sing}_{\mathfrak p}^*
\widehat S_r^{\rm hor}=\Omega_a,
\]
where $\mathfrak p$ ranges over finite continuation paths from the initial germ. For each $\Delta\in\Omega_a$, a suitable finite continuation path gives a nonzero logarithmic singularity at $\Delta$.

Since $\Omega_a\cap\R=\varnothing$, the positive real axis contains no nonzero Borel singularities. Hence for $t>0$,
\begin{equation}\label{eq:hhorizontalordinarysum}
e^{a/t}t^{3/2}p_t(r,0)
=\frac1t\int_0^\infty e^{-u/t}\widehat S_r^{\rm hor}(u)\dd u.
\end{equation}
In particular, the horizontal minimal action series has exact Gevrey order one and zero radius of convergence.
\end{theorem}

\begin{proof}
%The density $\rho_r^{\rm hor}$ in \eqref{eq:hhorizontalLaplace} is the single-branch pushforward of
%\[\frac{1}{2\pi^2}\frac{\tau}{\sinh\tau}\dd\tau\]
%under $F(-i\tau)-a$. Proposition \ref{prop:hhorizontalspectrum} therefore verifies the local-finiteness and no-escape hypotheses of Proposition~\ref{prop:hcontinuationAbel}. 
Apply Proposition~\ref{prop:hcontinuationAbel} to
\[
g(\tau)=F(-i\tau)-a,\qquad
\varpi(\tau)\dd\tau=
\frac{1}{2\pi^2}\frac{\tau}{\sinh\tau}\dd\tau.
\]
Proposition~\ref{prop:hhorizontalspectrum} gives local finiteness and no escape; the quotient $s=\tau^2$ removes the even symmetry at the initial point. Every critical point is accessible, and the amplitude is nonzero there; hence each $\Delta\in\Omega_a$ produces a nonzero square-root singularity of $\rho_r^{\rm hor}$. Lemma \ref{lem:hAbelsingular} converts it into a logarithmic singularity of $\widehat S_r^{\rm hor}$, proving the full singular-support statement.

Along the positive axis the inverse branch $\tau_+(u)$ is real and increasing, and
\[
\rho_r^{\rm hor}(u)=O_a\bigl((1+u)e^{-u/a}\bigr).
\]
Thus the Abel--Laplace identity applied to \eqref{eq:hhorizontalLaplace} gives \eqref{eq:hhorizontalordinarysum}. Since the positive Borel ray is singularity-free, the two lateral sums coincide. 
\end{proof}

\subsection{Isotropic complex endpoints and limiting actions}
\label{sec:hgeometry}

For the vertical endpoint $(0,y)$, $y>0$, Theorem \ref{thm:hactionclass} gives
\[
\Sigma_H^{\rm fin}(0,y)=\{j\pi y:j\in\Z\setminus\{0\}\}.
\]
The Fourier kernel depends on the horizontal variable only through $\sigma=x_1^2+x_2^2$. Thus a nonzero isotropic vector $x_*$ with
\[
 x_*\cdot x_*=0
\]
has the same Fourier kernel as $x=0$, although the complex endpoint equations distinguish the two cases. 

\begin{theorem}
\label{thm:hisotropic}
Fix $x_*\in\C^2\setminus\{0\}$ with $x_*\cdot x_*=0$, and let $y>0$.

\begin{enumerate}[label=\textnormal{(\roman*)}]
\item No finite complex normal geodesic joins the identity to $(x_*,y)$.

\item For every $t>0$, the Fourier kernel extends holomorphically to a neighborhood of $(x_*,y)$ and
\begin{equation}\label{eq:hisotropickernel}
p_t^\C(x_*,y)
=\frac{1}{8t^2}\sech^2\!\left(\frac{\pi y}{2t}\right).
\end{equation}

\item The finite limits of normal-geodesic actions along endpoints tending to $(x_*,y)$ are exactly
\[
\{j\pi y:j\in\Z\setminus\{0\}\}.
\]
More precisely, if
\[
(x_n,y_n)\to(x_*,y),
 \qquad
\mathcal A_n\to\mathcal A_\infty\in\C,
\]
then for a unique $j\in\Z\setminus\{0\}$,
\[
\lambda_n\to j\pi, \qquad \mathcal A_\infty=j\pi y, \qquad \|v_n\|\to\infty.
\]
Writing
\[
\varepsilon_n=\lambda_n-j\pi,\qquad
\sigma_n=x_n\cdot x_n,
\]
we have
\begin{equation}\label{eq:hisotropicuniversalscaling}
 \varepsilon_n v_n\longrightarrow j\pi x_*,
 \qquad
 \frac{\sigma_n}{\varepsilon_n^2}
 \longrightarrow
 \frac{2y}{j\pi}.
\end{equation}
\end{enumerate}
\end{theorem}

\begin{proof}
\textnormal{(i)}. If $\lambda\notin\pi\Z\setminus\{0\}$, the endpoint formula \eqref{eq:hvfull} and $x_*\cdot x_*=0$ imply $v\cdot v=0$, hence $y=0$. On the nonzero lattice $\lambda=j\pi$, the horizontal endpoint is zero. 

\textnormal{(ii)}. For fixed $t>0$, the Fourier integral is holomorphic in $(\sigma,y)$ near $(0,y)$ and depends on $x$ only through $\sigma=x\cdot x$. Setting $\sigma=0$ therefore gives the vertical formula \eqref{eq:hisotropickernel}.

\textnormal{(iii)}. For a finite trajectory with endpoint $(x_n,y_n)$ and action $\mathcal A_n$, the endpoint equations imply
\[
(\sigma_n^2+4y_n^2)\lambda_n^2-8\mathcal A_ny_n\lambda_n +4A_n^2-2\sigma_n\mathcal A_n=0.
\]
Hence, along a sequence with finite action limit and $(x_n,y_n)\to(x_*,y)$, the parameters $\lambda_n$ remain bounded and satisfy $y_n\lambda_n-\mathcal A_n\to0$. A limit away from $\pi\mathbb Z\setminus\{0\}$ would force
\[
y_n=\frac{\sigma_n}{2}\mathscr H(\lambda_n)\longrightarrow0,
\]
contradicting $y_n\to y>0$. Hence there is $j\in\mathbb Z\setminus\{0\}$ such that
\[
\lambda_n\to j\pi,\qquad \mathcal A_n\to j\pi y.
\]
Writing $\lambda_n=j\pi+\varepsilon_n$, the horizontal inversion formula and $\varepsilon_n^2\mathscr H(j\pi+\varepsilon_n)\to j\pi$ give
\[
\varepsilon_n v_n\to j\pi x_*,\qquad
\frac{\sigma_n}{\varepsilon_n^2}\to\frac{2y}{j\pi},
\]
and hence $\|v_n\|\to\infty$.

Conversely, fix $j\ne0$ and choose $w\in\mathbb C^2$ with
$x_*\cdot w=y$. Set
\[
\lambda_\varepsilon=j\pi+\varepsilon,
\qquad
v_\varepsilon=\frac{j\pi x_*}{\varepsilon}+\varepsilon w.
\]
The endpoint formulas then give
\[
(x_\varepsilon,y_\varepsilon)\to(x_*,y),
\qquad \mathcal A_\varepsilon\to j\pi y.
\]
Thus every lattice value occurs as a finite action limit.
\end{proof}

\section{Dimension-independent resurgence on isotropic Heisenberg groups}
\label{sec:hhigherdim}
On $H^n$, the phase $F$ is unchanged, while the Fourier amplitude becomes $V^n$. Let
\[
\mathcal L_n=\frac12\sum_{j=1}^{2n}X_j^2
\]
be the standard isotropic sub-Laplacian on $H^n$ and write $r^2=\sum_{j=1}^{2n}x_j^2$, $a=r^2/2$.

\begin{theorem}
\label{thm:hhigherdim}
For every fixed $n\ge1$,
\[
p_t^{(n)}(r,y)=\frac{1}{(2\pi t)^{n+1}}
\int_{\R}e^{-F(-i\tau)/t}V(\tau)^n\dd\tau.
\]
For real endpoints $r,y>0$, the finite action spectrum and the continued Borel singular support are independent of $n$:
\begin{align*}
 \operatorname{Sing}_{\Gamma_\pm}\widehat S_n
 &=\bigl(\Sigma_H^{\rm fin}(r,y)-A_1\bigr)\cap(0,\infty),\\
 \bigcup_{\mathfrak p}\operatorname{Sing}_{\mathfrak p}^*\widehat S_n
 &=\bigl(\Sigma_H^{\rm fin}(r,y)-A_1\bigr)\setminus\{0\}.
\end{align*}
The Morse and cubic ordinary Borel singularities remain logarithmic and of exponent $-1/6$, respectively. In the integral-cycle normalization, the directional Stokes matrices $S_\theta$, the positive matrix $S_0$, and their action-labelled logarithms are independent of $n$.

Near the vertical axis, the dominant action $\omega_1=B_1-A_1$ is independent of the dimension; the prefactor depends on $n$:
\[
[t^\ell]S_n(t)=\frac{\chi_n}{\sqrt\pi}\Gamma(\ell)\omega_1^{-\ell}
\bigl(1+O(\ell^{-1})\bigr),
\qquad\chi_n=-\frac{\sqrt2\,W(q_+)^n}{\sqrt{F''(q_+)}}.
\]
\end{theorem}

\begin{proof}
Central Fourier reduction changes only the amplitude from $V$ to $V^n$; the phase $F$ is unchanged. Hence the critical actions, inverse covering, no-escape estimates, and zero-cycle permutations are exactly those of $H^1$. Since $V^{n-1}$ is holomorphic and nonzero at every finite critical point, it only rescales the local saddle amplitudes and does not change their singular support or Stokes incidence. The singular-support and Stokes statements therefore follow from Sections~\ref{sec:hglobal} and \ref{sec:hglobalstokes}; replacing $W$ by $W^n$ in the first-pair calculation gives the large-order prefactor.
\end{proof}

\begin{corollary}
\label{prop:hhighergeometry}
For every $n\ge1$:
\begin{enumerate}[label=\textnormal{(\roman*)}]
\item At horizontal endpoints $r>0$, $y=0$, the positive Borel ray is singularity-free and positive Borel summation recovers $p_t^{(n)}(r,0)$ exactly.
\item If $x_*\in\C^{2n}\setminus\{0\}$ is isotropic and $y>0$, then no finite complex normal geodesic reaches $(x_*,y)$, while
\[
p_t^{(n),\C}(x_*,y)=p_t^{(n)}(0,y).
\]
The finite limiting action spectrum is
\[
\{j\pi y:j\in\Z\setminus\{0\}\},
\]
with the same covector scaling as in
\eqref{eq:hisotropicuniversalscaling}.
\end{enumerate}
\end{corollary}

\begin{proof}
For $y=0$, the quotient phase is independent of $n$, while the additional factor $V^{n-1}$ is holomorphic and nonzero at every finite critical point and improves the decay on the positive ray. Hence Theorem \ref{thm:hhorizontalresurgence} applies unchanged.

For isotropic endpoints, the Hamilton equations split into $n$ identical rotating pairs and depend on the horizontal variables only through $v\cdot v$ and $x\cdot x$. The endpoint formulas are the same with $x,v\in\C^{2n}$, so the proof of Theorem~\ref{thm:hisotropic} applies.
\end{proof}

\section{Conclusions and further questions}

For real endpoints, we identified finite complex normal actions with phase critical values and their differences with the finite singularities of the continued Borel germ. At nonresonant endpoints, the saddle sectors carry row-finite directional Stokes operators, with an explicit matrix and logarithm on the positive ray. The nearest singularity gives two-term large-order asymptotics near the vertical axis, and uniform confluence formulas describe the cubic resonances.

The original Fourier contour selects a particular combination of resurgent sectors. Positive Borel summation is exact at horizontal endpoints. Nonzero isotropic complex endpoints admit no finite normal trajectories, but have the signed vertical action lattice as their complete limiting spectrum. These action and Stokes structures extend unchanged to higher-dimensional isotropic Heisenberg groups.

\vskip 0.1cm
Two questions remain.
\begin{enumerate}[label=\textnormal{(\arabic*)}]
\item Can the completed Stokes representation be extended uniformly through cubic resonances and to angular sectors of aperture at least $\pi$?

\item Can one determine the nearest Borel singularities, and hence precise large-order asymptotics, for general endpoints beyond the near-vertical regime?
\end{enumerate}

\end{document}